%% file: mu-calculus.tex
\documentclass{llncs}

\usepackage{times,tikz}
\usepackage{xspace}
\usepackage{diagrams}
\diagramstyle[centredisplay,inline]

\def\mucalc{$\mu$-calculus\xspace}
\def\muknl{$\mu$-calculus of knowledge\xspace}
\def\munonmix{$\mu$-calculus of non-mixing epistemic fixpoints\xspace}

\usepackage{amsmath}
\usepackage{amssymb}
\usepackage{graphics,graphicx}
\usepackage{macros_art}
\usepackage{myMnSymbol}

\renewcommand{\sdpart}{\rightharpoonup}

\def\supp{{ \mathsf{supp}}}

\def\lfp{\mathsf{lfp}}
\def\gfp{\mathsf{gfp}}

\def\outdegree{\mathsf{outdeg}}

\def\ZZZ{\mathcal{Z}}

\def\Runs{{ \mathsf{Runs}}}

\newtheorem{rem}[theorem]{\textbf{Remark}}

\title{Model-checking an Epistemic $\mu$-calculus with Synchronous and Perfect Recall Semantics}
\titlerunning{Model-checking Epistemic $\mu$-calculus}

\begin{document}

\author{Rodica Bozianu\inst{1,3} \and C\u at\u alin Dima\inst{1} \and Constantin Enea\inst{2} }

\institute{
LACL, Universit\'{e} Paris Est-Cr\'{e}teil, \\
61 av. du G-ral de Gaulle,
94010 Cr\'eteil, France
\and
LIAFA, CNRS UMR 7089, 
Universit\'{e} Paris Diderot - Paris 7, \\
Case 7014, 
75205 Paris Cedex 13, France
\and
Faculty of Computer Science, "A.I.Cuza" University\\
700483 Ia\c{s}i, Romania
}

\maketitle

\input{introduction.tex}

\vspace{-1mm}
\section{Preliminaries}

We start by fixing a series of notions and notations used in the rest of the paper.

$A^*$ denotes the set of words over $A$.
The length of $\alpha \in A^*$, is denoted $|\alpha|$
and the prefix of $\alpha$ up to position $i$ is denoted $\alpha[1..i]$.
Hence, $\alpha[1..0] = \eps$ is the empty word.
The (strict) prefix ordering on $A^*$ is denoted $\preceq$ ($\prec$).


Given a set $A$ and an integer $n \in \Nset$, 
an \emph{$A$-tree of outdegree $\leq n$} is a partial function $t : [1\ldots n]^* \sdpart A$ 
whose support, denoted $\supp(t)$, is a prefix-closed subset of 
the finite sequences of integers in $[1\ldots n]$.
A \emph{node} of $t$ is an element of its support.
A \emph{path} in $t$ is a pair $(x,\rho)$ consisting of a node $x$ and 
the sequence of $t$-labels of all the nodes which are prefixes of $x$,
$\rho = \big(t(x[1\ldots i])\big)_{0\leq i \leq |x|}$.

A \textbf{multi-agent system} (\textbf{MAS}, for short) is a tuple $M = \big(Q,Ag,\delta,q_0,\Pi,(\Pi_a)_{a\in Ag},\pi\big)$
with $Ag$ being the set of agents, $Q$ the set of states, $q_0$ the initial state of the system, $\delta \subseteq Q \times Q$, 
$\Pi$ the set of \emph{atomic propositions}, $\pi : Q \sd 2^\Pi$ and for all $a \in Ag$, $\Pi_a \subseteq \Pi$.
A run in the structure $M$ from a state $q_0$ is an infinite sequence of states $\rho = q_0q_1q_2...$ such that $(q_i,q_{i+1}) \in \delta$ for all $i \geq 0$.
The set of finite runs in $M$ is denoted $\Runs(M)$.
Throughout this paper we consider only finite systems, with $Q = \{1,\ldots,n\}$ 
and $q_0 = 1$, and we assume that $Q$ contains only reachable states.\label{assumplion}

The $2^\Pi$-tree representing the \emph{unfolding} of a MAS $M$, denoted $t_M$, is defined by
$\supp(t_M) = \{ x \in \Nset^* \mid 1x \in \Runs(M)\} \text{ and } t_M(x) = x[|x|]$.
%
For any two positions $x,y \in \supp(t_M)$ with $|x| = |y|$,
we denote $x \sim_a y$ if for any $n\leq |x|$ we have that 
\vspace{-1mm}
$$
\pi(t(x[1..n]))\cap \Pi_a = \pi(t(y[1..n]))\cap \Pi_a
\vspace{-1.5mm}
$$
Henceforth, for a word $w \in (2^\Pi)^*$, by $w \restr{\Pi_a}$ we denote
the sequence defined by 
$w\restr{\Pi_a}[i] = w[i]\cap \Pi_a$ for each $1\leq i\leq |w|$.
%
Note also that the relation $x \sim_a x'$ is both a relation on 
the nodes of the tree $t_M$ and on the runs of $M$.


\vspace{.5mm}
\noindent
{\bf Predicate transformers:}
Given a set $A$, an \emph{$A$-transformer} is a mapping $f : \big(2^{A}\big)^n \sd 2^{A}$.



Following the Knaster-Tarski theorem, any monotone $A$-transformer $f : 2^{A} \sd 2^{A}$
has a unique least and greatest fixpoint, denoted $\lfp_f$, resp. $\gfp_f$.

For an $A$-transformer $f : \big(2^{A}\big)^n \sd 2^{A}$, 
a tuple of sets $B_1,\ldots, B_n \subseteq A$ 
and some $k \leq n$ 
we denote 
$f_k(B_1,\ldots,B_{k-1},\cdot,B_{k+1},\ldots, B_n)  : 2^A \sd 2^A$ 
the $A$-transformer with 
\vspace{-1mm}
$$
f_k(B_1,\ldots,B_{k-1},\cdot,B_{k+1},\ldots, B_n)(B) = f(B_1,\ldots,B_{k-1},B,B_{k+1},\ldots, B_n)
\vspace{-1.5mm}
$$
Note that when $f$ is monotone, $f_k(B_1,\ldots,B_{k-1},\cdot,B_{k+1},\ldots, B_n)$ is monotone too.
Hence, both $\lfp_{f_k(B_1,\ldots,B_{k-1},\cdot,B_{k+1},\ldots, B_n)}$
and $\gfp_{f_k(B_1,\ldots,B_{k-1},\cdot,B_{k+1},\ldots, B_n)}$ exist.
These fixpoints can also be seen as the following $A$-transformers:
$\lfp^k_{f} : (2^{A})^{n} \sd 2^{A}$ and 
$\gfp^k_{f} : (2^{A})^{n} \sd 2^{A}$, defined respectively as:
\vspace{-2mm}
$$\lfp^k_{f} (B_1,\ldots, B_{n})  = \lfp_{f_k(B_1,\ldots,B_{k-1},\cdot,B_{k+1},\ldots, B_{n})}$$
\vspace{-5mm} 
$$\gfp^k_{f} (B_1,\ldots, B_{n})  = \gfp_{f_k(B_1,\ldots,B_{k-1},\cdot,B_{k+1},\ldots, B_{n})}$$

Note that both these $A$-transformers are constant in their $k$-th argument.
It is also known that both these $A$-transformers are monotone if $f$ is monotone.



\vspace{-1mm}
\section{The $\mu$-calculus of Knowledge}


{\bf Syntax:}
The syntax of the \textbf{\muknl} (in positive form) is based on the following sets of 
symbols: a finite set of \emph{agents} $Ag$, 
a family of finite sets of \emph{atomic propositions} $(\Pi_a)_{a\in Ag}$ 
(no restrictions apply on the pairwise intersections between these sets),
with $\Pi = \bigcup_{a\in Ag} \Pi_a$, and a finite set of \emph{second-order variables} 
$\ZZZ = \{Z_1,\ldots,Z_k\}$.
The set $\Pi_a$ represents the set of atoms whose value is \emph{observable} by agent $a$
at each instant (in the sense to be developed further).

The grammar for the formulas of the \muknl is:
\vspace{-1mm}
$$\varphi ::=  p \mid \neg p \mid Z \mid \varphi \wedge \varphi \mid \varphi \vee \varphi
\mid AX \varphi \mid EX \varphi 
 \mid K_a \phi \mid P_a \phi \mid \mu Z. \varphi \mid \nu Z. \varphi $$

\vspace{-2mm}
\noindent
where $p \in \Pi$, $a\in Ag$ and $Z \in \ZZZ$.
Formulas of the type $K_a \phi$ are read as \emph{agent $a$ knows that $\phi$ holds}.
The dual of $K_a$, denoted $P_a$,
(and definable as $P_a \phi = \neg K_a \neg \phi$ if negation were allowed), 
reads as \emph{agent $a$ considers that $\phi$ is possible}.
As usual, for a subset of agents $A \subseteq Ag$
we may denote $E_A$ the ``everybody knows'' operator, 
$E_A \phi = \bigwedge_{a \in A} K_a\phi$.

The fragment of the \muknl which does not involve the knowledge operator $K_a$ (or its dual)
is called here the \emph{plain \mucalc,} or simply the \mucalc, when there's no 
risk of confusion.
As usual, we say that a formula $\phi$ is \emph{closed} if 
each variable $Z$ in $\phi$ occurs in the scope of a fixpoint operator 
for $Z$.

We will also briefly consider in this paper 
the \emph{modal $\mu$-calculus of knowledge}, for the sake of comparison with 
other combinations of temporal and epistemic logics.
It has almost the same grammar, but with the nexttime operators $EX$ and $AX$ replaced 
with modal nexttime operators $\langle \br\alpha \rangle$, 
resp. $[\alpha]$ with $\br\alpha$ representing
a tuple of action symbols $\br\alpha = (\alpha_a)_{a\in Ag}$.
Note that the modal \muknl can be translated to the non-modal \muknl
by converting each action name $\alpha\in Act_a$ 
into an atomic proposition,
so the main results of this paper generalize easily to this calculus.



We give two semantics of the \muknl: a tree semantics and a finitary semantics. 
The tree semantics is required because we assume that agents have \emph{perfect recall},
and hence they remember all observations made since the system started.
The second is necessary for the decision problem. 
The equivalence between the two semantics on trees generated by MASs, 
which gives the decidability of the model-checking problem,
is a weak form of \emph{memoryless determinacy for tree automata}.
We present here both semantics of the \muknl in a predicate-transformer flavor, more appropriate
for stating a number of properties of the logic.

\vspace{.5mm}
\noindent
{\bf The tree semantics}
of the \muknl is given in terms of $2^{\Pi\cup \ZZZ}$-trees.
For a given tree $t$, 
each formula $\phi$ which contains variables $Z_1,\ldots, Z_n$
is associated with 
a $\supp(t)$-transformer $\|\phi\| : \big(2^{\supp(t)}\big)^n \sd 2^{\supp(t)}$ by structural induction, as follows:
\vspace{-1mm}
\begin{itemize}
\item The two atoms $p$ and $\neg p$ are interpreted as constant $\supp(t)$-transformers
$\|p\|: \big(2^{\supp(t)}\big)^n \sd 2^{\supp(t)}$ and $\|\neg p\| : \big(2^{\supp(t)}\big)^n \sd 2^{\supp(t)}$, 
defined by the sets $\|p\| = \{ x \in supp(t) \mid p \in \pi(t(x)) \}$,
resp. $\|\neg p\| = \{ x \in supp(t) \mid p \not \in \pi(t(x)) \}$.
\item Each variable $Z_i \in \ZZZ$ is interpreted 
as the $i$-th projection on $\big(2^{\supp(t)}\big)^n $, that is, 
as the $\supp(t)$-transformer
$\|Z_i\| : \big(2^{\supp(t)}\big)^n \sd 2^{\supp(t)}$ with 
$
\|Z_i\| ( S_1,\ldots,S_n)  = S_i, \forall S_1,\ldots,S_n \subseteq \supp(t)
$.

\item $\| \phi_1 \vee \phi_2\| = \|\phi_1\| \cup \|\phi_2\|$
and $\| \phi_1 \wedge \phi_2\| = \|\phi_1\| \cap \|\phi_2\|$.
\item Each of the two nexttime operators is mapped to a $\supp(t)$-transformer, denoted
$AX$, resp. $EX : 2^{\supp(t)} \sd 2^{\supp(t)}$, defined as follows: for each $S \subseteq \supp(t)$,
 \vspace{-1mm}
\begin{align*}
AX(S) & = \{ x \in \supp(t) \mid \forall i \in \Nset \text{ if } xi \in \supp(t) 
  \text{ then } xi \in S \} \\
EX(S) & = \{ x \in \supp(t) \mid \exists i \in \Nset \text{ with } xi \in \supp(t) 
 \text{ and }  xi \in S \}
\end{align*}

  \vspace{-2mm}
  \noindent
Then $\| AX\phi\| = AX \circ \|\phi\|$, 
and $\| EX\phi\| = EX \circ \|\phi\|$.
\item Each pair of epistemic operators is mapped to
$\supp(t)$-transformers $K_a$, resp. $P_a : 2^{\supp(t)} \sd 2^{\supp(t)}$,
defined as follows: for each $S \subseteq \supp(t)$,
 \vspace{-1mm}
\begin{align*}
K_a(S) & = \{ x \in \supp(t) \mid \forall y\in \supp(t) \text{ with } x \sim_a y 
 \text{ we have } y\in S \} \\
P_a(S) & = \{ x \in \supp(t) \mid \exists y\in \supp(t) \text{ with } x \sim_a y 
 \text{ and } y\in S \}
\end{align*}

\vspace{-2mm}
  \noindent
Then $\|K_a \phi\| = K_a \circ \|\phi\|$ and 
$\|P_a \phi\| = P_a \circ \|\phi\|$.
\item For the fixpoint operators we put $\|\mu Z_i. \phi\| = \lfp^i_{\|\phi\|}$
and $\|\nu Z_i. \phi\| = \gfp^i_{\|\phi\|}$.
\end{itemize}

\vspace{-1mm}
\noindent
Note that the two $supp(t)$-transformers $K_a$ and $P_a$ are dual and we have that $K_a(S) = \br{P_a(\br S)}$, with $\br \cdot$ denoting the set complementation.
{\em We also denote $t \models \phi$ iff $\eps \in \|\phi\|$.}



The following property says that the \muknl cannot distinguish between 
isomorphic trees:
\begin{proposition}\label{prop:tree-iso}
For any two MASs $M_1$ and $M_2$ for which there exists some 
tree isomorphism 
$\chi : 2^{supp(t_{M_1})} \sd 2^{supp(t_{M_2})}$,
and for any \muknl formula $\phi$,
the following diagram commutes:
\vspace*{-7pt}

\begin{equation}
\begin{diagram}[tight,height=20pt,width=4em]
\big(2^{\supp(t_1)}\big)^n & \rTo^{\|\phi \|_{M_1}} & 2^{\supp(t_2)} \\
\dTo^{(\chi)^n} & & \dTo_{\chi} \\
\big(2^{\supp(t_2)}\big)^n & \rTo^{\|\phi\|_{M_2}} & 2^{\supp(t_2)} \\
\end{diagram} 
\label{diagram:finite} 
\vspace*{-6pt}
\end{equation}
\end{proposition}
\input{proof-prop7.tex}


\vspace{-2mm}
\subsection{Comparison with other temporal epistemic frameworks}
We discuss the relationship between the \muknl and 
other temporal epistemic logics or game models with imperfect information and perfect recall.

As already noted e.g. in \cite{shilov-garanina}, 
the following fixpoint formula defines the common knowledge operator for two agents:
$C_{a,b} \phi = \nu Z . (\phi \wedge K_a Z \wedge K_b Z )$.

On the other hand, it's easy to see that the (modal variant of the) \muknl is more expressive than 
the alternating epistemic $\mu$-calculus of \cite{bulling-jamroga-mu}, 
due to the possibility to insert knowledge operators ``in between'' the 
quantifiers that occur in the semantics of the coalition operators.
The relationship with $ATL_{iR}$ is more involved, as we detail in the sequel.

Given a set of agents $A \subseteq Ag$, denote $Act_A$ the cartesian product 
of the set of action symbols for each agent in $A$, $Act_A = \bigtimes_{a\in A} Act_a$. 
Then, formulas of the type 
$\llangle A \rrangle \Box p$ can be expressed as the fixpoint formula 
$\ds\nu Z. \bigvee_{\alpha \in Act_a} K_a \big( p \wedge \!\!\!\bigwedge_{ \beta \in Act_{Ag\setminus \{a\}}} \!\!\![\alpha, \beta] Z \big)$.

\label{atl}
Formulas containing the until operator cannot be translated into the \muknl.
The reason is similar to the one explained in 
\cite{bulling-jamroga-mu}: 
in formulas of the type $\llangle a \rrangle \Diamond p$
the objective $p$ might not be observable by the agent $a$,
who might only be able to know that, at some given time instance, 
sometimes in the past, the objective was achieved on all identically observable traces.
\input{ATL_MuCalc.tex}

This construction can be extended to the whole $ATL$ by structural induction on the formula.\\





\label{rema:games}
Multi-player games with incomplete information can also be translated into the \muknl. 
Recall briefly that a (synchronous) two-player game is a tuple \\
$G = \Big(Q, Ag, (Act_a)_{a\in Ag}, \delta, Q_0, (Obs_a)_{a\in Ag}, (o_a)_{a\in Ag}, par\Big)$ with 
$Q$ denoting the set of states, $Ag=\{A,B\}$ the set of players, 
$\delta \subseteq Q \times \bigtimes_{a\in Ag} Act_a \times Q$ denoting the transition relation,
$o_a : Q \sd Obs_a$ denoting the observability relation for player $a$ and
$par : Q\sd \Nset$ defining the \emph{parity} of each state.



A player $a\in Ag$ plays by choosing a \emph{feasible strategy},
which is a mapping $\sigma : (Obs_a)^* \sd Act_a$.
A strategy for $a$ is \emph{winning} 
when all the runs that are compatible with that strategy satisfy the property: 
the maximal parity of a state which occurs infinitely often in the run is even.
The winning condition might be non-observable to  $a$,
as it might happen that two identically observable states $q_1,q_2 \in Q$ 
might have different parities.

The set of winning strategies for a player in a multi-player game with imperfect information is then expressible within the \muknl,
similarly to the known encoding of the set of winning strategies in a parity game into the \mucalc 
from e.g. \cite{emerson-jutla91,obdrzalek-thesis}. 
Assuming that the largest parity in $Q$ is even and
the atomic proposition $p_i$ holds exactly in all states with parity $i$,
the following \muknl formula encodes the winning strategies for player $a$:
 \vspace{-1mm}
$$
\nu Z_n \mu Z_{n-1} \ldots \mu Z_1 .
\bigvee_{\alpha \in Act_a} K_a  \bigvee_{i\leq n} \big (p_i \wedge \bigwedge_{\beta \in Act_{Ag\setminus \{a\}}} [ \alpha,\beta]  Z_i\big) 
 \vspace{-1.5mm}
$$



\subsection{The model-checking problem}

The model-checking problem for the \muknl is the problem of 
deciding, given a MAS $M$ and a closed formula $\phi$, whether $t_M \models \phi$.

The undecidability of the model-checking problem for combinations of temporal and epistemic logics
based on a synchronous and perfect recall semantics and containing the common knowledge operator 
\cite{meyden-shilov,van-benthem-pacuit} implies the following  result.

\begin{theorem}
The model-checking problem for the \muknl 
is undecidable.
\end{theorem}

The next two sections are dedicated to finding 
a fragment of the \muknl with a decidable model-checking problem.

\section{Revisiting the Decidability of the Model-checking Problem for the Tree Semantics of the plain
$\mu$-calculus}


Given a multi-agent system $M = (Q,Ag,\delta,q_0,\Pi,(\Pi_a)_{a\in Ag},\pi)$,
and an agent $a \in Ag$,
we may define the relation $\Gamma_a^M \subseteq Q \times Q$\label{def:Gamma} as follows: 
$(q,r) \in \Gamma_a^M$ if for any run $\rho$ in $M$ ending in $q$ (i.e. $\rho[|\rho|] = q$) there exists a run $\rho'$ ending in $r$ with 
$\rho \sim_a \rho'$. 
Whenever the MAS M is understood from the context, we use the notation $\Gamma_a$ instead of $\Gamma_a^M$.

We now define a second semantics for the \muknl, which works on the \emph{set of states} of a 
MAS $M$.
Each formula $\phi$ which contains variables $Z_1,\ldots, Z_n$
is associated with 
a $Q$-transformer $\lceil\phi\rceil  : \big(2^{Q}\big)^n \sd 2^{Q}$, again by structural induction:
\begin{itemize}
\item $\lceil p\rceil$ resp. $\lceil \neg p \rceil$  
are the constant $Q$-transformers
$\lceil p\rceil = \{ q \in Q \mid p \in \pi(q) \}$, resp. 
$\lceil \neg p \rceil = \{ q \in Q \mid p \not \in \pi(q) \}$.
\item $\lceil Z_i\rceil  : \big(2^{Q}\big)^n \sd 2^{Q}$ is the $i$-th projection, i.e., 
given $S_1,\ldots,S_n \subseteq Q$, 
$\lceil Z_i\rceil  ( S_1,\ldots,S_n)  = S_i$.
\item $\lceil  \phi_1 \vee \phi_2\rceil  = \lceil \phi_1\rceil  \cup \lceil \phi_2\rceil $,
and $\lceil  \phi_1 \wedge \phi_2\rceil  = \lceil \phi_1\rceil  \cap \lceil \phi_2\rceil $.
\item Both nexttime modalities are associated with $Q$-transformers 
$AX^f, EX^f : 2^{Q} \sd 2^{Q}$ defined as:
\begin{align*}
\vspace*{-17pt}
AX^f(S) & = \{ q \in Q \mid \forall r \in Q \text{ if } (q, r) \in \delta 
\text{ then } r \in S \} \\
EX^f(S) & = \{ q \in Q \mid \exists r \in Q \text{ with } (q, r) \in \delta 
\text{ and } r \in S \}
\end{align*}

\vspace{-1.5mm}
\noindent
Then $\lceil  AX\phi\rceil  = AX^f \circ \lceil \phi\rceil $ and, 
similarly, 
$\lceil  EX\phi\rceil  = EX^f \circ \lceil \phi\rceil $,
\item Both epistemic operators are associated with $Q$-transformers
$K_a^{f}, P_a^{f} : 2^{Q} \sd 2^{Q}$ 
defined as:
\vspace*{-5pt}
\begin{align*}
K_a^{f}(S) = \br{\Gamma_a(\br S)} = & \{ q \in Q \mid \forall s \in Q,\text{ if }(s,q)\in \Gamma_a \text{ then } s\in S\} \\
P_a^{f}(S) = \Gamma_a(S) = & \{ q \in Q \mid \exists s \in S \text{ s.t. } (s,q) \in \Gamma_a \} 
\end{align*}

\vspace{-1.5mm}
\noindent
Then $\lceil P_a \phi\rceil  = P_a^{f} \circ \lceil \phi\rceil $ 
and $\lceil K_a \phi\rceil  = K_a^{f} \circ \lceil \phi\rceil $.




\item $\lceil \mu Z_i. \phi\rceil  = \lfp^i_{\lceil \phi\rceil }$
and $\lceil \nu Z_i. \phi\rceil  = \gfp^i_{\lceil \phi\rceil }$.
\end{itemize}
\noindent
The following result
represents a variant of the Finite Model Theorem for \mucalc and 
is proved by structural induction on the formula $\phi$ in 
\cite{bozianu-dima-enea-arxiv}:

\begin{theorem}\label{teo:plain}
Given a MAS $M = (Q,Ag,\delta,q_0,\Pi,(\Pi_a)_{a\in Ag},\pi)$ 
in which $Q = \{1,\ldots,n\}$ and $q_0=1$, and a (plain) \mucalc formula $\phi$, 
the following diagram commutes:
\begin{equation}
\begin{diagram}[tight,height=20pt,width=4em]
\big(2^Q\big)^n & \rTo^{\lceil \phi \rceil} & 2^Q \\
\dTo^{(t_M^{-1})^n} & & \dTo_{t_M^{-1}} \\
\big(2^{\supp(t)}\big)^n & \rTo^{\|\phi\|} & 2^{\supp(t)} \\
\end{diagram}\label{diagram:plain}  
\vspace*{-4.5mm}
\end{equation}
\medskip
\end{theorem}
We also say that the diagram \ref{diagram:plain} holds (or commutes) for the formula $\phi$ 
in the system $M$.

\input{theorem2.tex}




\input{mu-nonmix.tex}

\vspace{-2mm}
The following result follows from a similar result for
LTLK from \cite{meyden-shilov}.
A self-contained proof can be found in \cite{bozianu-dima-enea-arxiv}:

\vspace{-1.5mm}
\begin{theorem}
The model checking problem for the \munonmix is hard for non-elementary time.
\vspace{-3mm}
\end{theorem}

\input{conclusions.tex}

\bibliographystyle{plain}
\bibliography{lics-bibtex}
\nocite{HalpernVardi86,lomuscio-mcmas,kacprzak-penczek-aamas-2005,bulling-jamroga-mu,goranko-drimmelen06,agotnes-synthese06,arnold-rudiments,walukiewicz-tcs-mu-calculus}

\end{document}

%% file: introduction.tex
\vspace*{-20pt}
\begin{abstract}
We show that the model-checking problem is decidable for a fragment 
of the epistemic $\mu$-calculus with imperfect information and perfect recall.
The fragment allows free variables within the scope of epistemic modalities in a restricted 
form that avoids constructing formulas embodying any form of common knowledge.
Our calculus subsumes known decidable fragments of 
epistemic $CTL/LTL$, may express winning strategies in 
two-player games with one player having imperfect information and 
non-observable objectives,
and, with a suitable encoding, 
decidable instances of the model-checking problem for $ATL_{iR}$ can be 
encoded as instances of the model-checking problem for the \muknl.
\end{abstract}

\vspace*{-10pt}
\section{Introduction}

The $\mu$-calculus of knowledge is an enrichment of the $\mu$-calculus on trees
with individual epistemic modalities $K_a$ 
(and its dual, denoted $P_a$).
It is designed with the aim that, like the classical modal $\mu$-calculus, 
it would subsume most combinations of temporal and epistemic logics.
The \muknl is more expressive than linear or branching temporal epistemic logics
\cite{HalpernVardi86,shilov-garanina}, propositional dynamic epistemic logics \cite{van-benthem-pacuit},
or the alternating epistemic $\mu$-calculus \cite{bulling-jamroga-mu}. 
On the other hand, some gaps in its expressive power seem to exist, 
as witnessed by recent observations in \cite{bulling-jamroga-mu}
showing that formulas like $\llangle a \rrangle p_1 \UUU p_2$
are not expressible in the fixpoint version of $ATL$.
This expressivity gap can be reproduced in the \muknl, 
though the \muknl is richer than the alternating $\mu$-calculus.



A rather straightforward fragment of the epistemic $\mu$-calculus which 
has a decidable model-checking problem is the one 
in which knowledge modalities apply only to closed formulas,
that is, formulas in which all second-order variables are bound by some fixpoint operator.
The decidability of this fragment follows from recent results 
on the decidability of the emptiness problem for two player games with 
one player having incomplete information and with non-observable winning conditions \cite{chatterjee-doyen-fullpogames}.

However more expressive fragments having a decidable model-checking problem seem to exist.
For example, winning strategies in two-player games with imperfect information can be encoded as 
fixpoint formulas in the \muknl, but not in the above-mentioned restricted fragment.
The same holds for some formulas in 
$ATL$ with imperfect information and perfect recall ($ATL_{iR}$) \cite{schobbens-atl-ir2004,bulling-atl-survey}:
the $ATL$ formula 
$\llangle a \rrangle \Box p$
can be expressed in a modal $\mu$-calculus of knowledge as 
$\ds\nu Z. \bigvee_{\alpha \in Act_a} K_a \big( p \wedge \!\!\!\bigwedge_{ \beta \in Act_{Ag\setminus \{a\}}} \!\!\![\alpha, \beta] Z \big)$.
And there are variants of $ATL_{iR}$ 
for which the model-checking problem is decidable \cite{dima-gandalf}.
Note that a translation of each instance of the model-checking problem
for $ATL$ into instances of the model-checking problems for the \muknl 
is also possible but requires the modification of the models, as suggested on page \pageref{atl} below.



Our aim in this paper is to identify such a larger fragment of the epistemic $\mu$-calculus for which 
model-checking is decidable.
The fragment we propose here allows an epistemic modality $K_a$ to be applied to 
a non-closed $\mu$-calculus formula $\phi$, but in such a way that avoids 
expressing properties that construct any variant of common knowledge for two or more agents.
Roughly, the technical restriction is the following: 
two epistemic operators, referring to the knowledge of two different agents $a$ and $b$,
can be applied to non-closed parts of a formula only if the two agents 
have \emph{compatible} observations
(in the sense that the observability relation of one of the agents is a refinement of the observability relation of 
the other agent).
The variant presented here relies on a \emph{concrete} semantics, in the sense of \cite{dima-jlc},
with the observability relation for each agent $a$ being syntactically identified by a subset $\Pi_a$ 
of atomic propositions. 
We require this in order to syntactically define our 
fragment of \muknl with a decidable model-checking problem: 
the compatibility of two observability relations $\sim_a$ and $\sim_b$  
is specified at the syntactic level by imposing that either $\Pi_a \subseteq \Pi_b$ or vice-versa.

The epistemic $\mu$-calculus with perfect recall has a history-based 
semantics: for each finite transition system $T$, 
the formulas of the epistemic $\mu$-calculus must be interpreted over 
the \emph{tree unfolding} of $T$.
This makes it closer with the tree interpretations of the $\mu$-calculus 
from \cite{emerson-jutla91}.
For the classical $\mu$-calculus, there are two ways of proving that the satisfiability and the model-checking 
problem for the tree interpretation of the logic is decidable:
either by providing translations to parity games, or by 
means of a Finite Model Theorem which ensures that 
a formula has a tree interpretation iff it has a \emph{state-based} interpretation
over a finite transition system (this is known 
to be equivalent with memoryless determinacy for parity games, see e.g. \cite{bradfield-modal}).

The generalization of the automata approach does not seem to be possible 
for epistemic $\mu$-calculus, 
mainly due to the absence of an appropriate generalization of tree automata equivalent with the \muknl.
So we take the approach of providing a generalization of the Finite Model Theorem for 
our fragment of the epistemic $\mu$-calculus.
This result
says roughly that the tree interpretation of a 
formula over the tree unfolding of a given finite transition system $T$
which contains the epistemic operators $K_a$ or $P_a$ is exactly the ``tree unfolding'' of
the finitary interpretation of the formula in a second transition system $T'$,
which is obtained by determinizing the projection of $T$ 
onto the observations of agent $a$, 
a construction that is common for decidable fragments of temporal epistemic logics.
Our contribution consists in showing that this construction 
can be applied for the appropriate fragment of the \muknl.
The proof is given in terms of commutative diagramms between 
predicate transformers that are the interpretations of non-closed formulas. 



The model checking problem for the decidable fragment of the epistemic $\mu$-calculus 
is non-elementary hard due to the non-elementary hardness of the model-checking problem for 
the linear temporal logic of knowledge \cite{meyden-shilov}.
In the full version of this paper \cite{bozianu-dima-enea-arxiv}, 
we provide a self-contained proof of this result, by 
a reduction of the emptiness 
problem for star-free regular expressions. 


The rest of the paper is divided as follows:
in the next section we recall the predicate transformer semantics of the $\mu$-calculus and 
adapt it to our epistemic extension, both for the tree interpretation and the finitary interpretation. 
We then give our weak variant of the Finite Model Theorem for the classical $\mu$-calculus in the third section.
The fourth section serves for introducing our fragment of the epistemic $\mu$-calculus
and for proving the decidability of its model-checking problem.
We end with a section with conclusions and comments.

%

%% file: proof-prop7.tex
\begin{proof}
By straightforward structural induction on the formula $\phi$.

Let $S_1,...,S_n \subseteq 2^{supp(t_1)}$. We have to prove that $\chi(\|\phi \|_{M_1}(S_1,...,S_n)) = \| \phi \|_{M_2} (\chi(S_1),...,\chi(S_n))$.
\begin{enumerate}
	\item For $\phi = p$ we have that 
	\begin{align}
	\chi (\|\phi & \|_{M_1}(S_1,...,S_n)) = \chi(\{ x\in t_{M_1} \mid \phi \in t_{M_1}(x)\}) \tag*{}\\
 		&= \{ \chi(x) \mid x \in t_{M_1}, \phi \in t_{M_1}(x)\}  \tag*{}\\
 		&= \{ y \in t_{M_2} \mid \exists x \in t_{M_1} \text{ with } \phi \in t_{M_1}(x) s.t. \chi (x) = y \} \tag*{ since $\chi$  is a bijection} \\
	 &= \{ y \in t_{M_2} \mid \phi \in t_{M_2}(y)\} \tag*{}\\
	 &= \| \phi \|_{M_2} (\chi (S_1),...,\chi (S_n)) \tag*{}
	\end{align}
	  The proof is similar for $\phi = \neg p$.
	
	\item For $\phi = Z_i \in \ZZZ $, $\| \phi \|_{M_1} (S_1,...S_n) = S_i$. Then,
	\begin{align*}
	\chi(\|\phi \|_{M_1}&(S_1,...,S_n)) = \chi(S_i) 
	 = \| \phi \|_{M_2} (\chi(S_1),...,\chi(S_n)).
	\end{align*}
	
	\item For $\phi = \phi_1 \vee \phi_2$, we have $\| \phi \| = \| \phi_1 \| \cup \| \phi_2 \|$. 
	By assuming that the property holds for $\phi_1$ and $\phi_2$, we get 
	\begin{align*}
	\chi(\| \phi \|_{M_1}&(S_1,...,S_n)) = \chi(\| \phi_1 \|_{M_1}(S_1,...,S_n) \cup \| \phi_2 \|_{M_1}(S_1,...,S_n)) \\ 				& = \chi(\| \phi_1 \|_{M_1}(S_1,...,S_n)) \cup \chi(\| \phi_2 \|_{M_1}(S_1,...,S_n)) \\
	& = \| \phi_1 \|_{M_2}(\chi(S_1),...,\chi(S_n)) \cup \| \phi_2 \|_{M_2}(\chi(S_1),...,\chi(S_n)) \\
	& = \| \phi_1 \vee \phi_2 \|_{M_2}(\chi(S_1),...,\chi(S_n)).
	\end{align*}
	 We similar proof can be given for $\phi = \phi_1 \wedge \phi_2$.
	
	\item For $\phi = AX \phi_1$, $\| \phi \| = AX \circ \| \phi_1 \|$. We have that 
	\begin{align}
	\chi &\big (\| AX \phi_1 \|_{M_1}(S_1,...,S_n)\big) = \chi \Big(AX \big(\| \phi_1 \|_{M_1}(S_1,...,S_n)\big)\Big) \tag*{}\\
	&= \{\chi(x) \mid x \in \supp(t_{M_1}) \text{ and } \forall i \in \Nset, \text{ if } xi \in \supp(t_{M_1}) \tag*{}\\
		&  \qquad \qquad \text{ then } t_{M_1}(xi) \in \|\phi_1\|_{M_1}(S_1,...,S_n) \} \tag*{}\\ 
	&= \{ y \mid \chi^{-1}(y) \in \supp(t_{M_1}) \text{ and } \forall j \in \Nset, \text{ if } \chi^{-1}(y)\chi^{-1}(j) \in \supp(t_{M_1}),\tag*{}\\
			&  \qquad \qquad \text{ then } t_{M_1}(\chi^{-1}(y)\chi^{-1}(j)) \in \| \phi_1 \|_{M_1}(S_1,...,S_n)\} \tag*{since $\chi$ is bijective}\\
	 &= \{ y \mid y \in \supp(t_{M_2}) \text{ and } \forall j \in \Nset, \text{ if } yj \in \supp(t_{M_2}) \tag*{}\\
	 		&  \qquad \qquad \text{ then } t_{M_2}(yj) \in \chi(\| \phi_1\|_{M_1}(S_1,...,S_n))\} \tag*{}\\
	 &= \{ y \in \supp(t_{M_2}) \mid \forall j \in \Nset \text{ if } yj \in \supp(t_{M_2}) \tag*{}\\
	 		&  \qquad \qquad \text{ then } t_{M_2}(yj) \in \| \phi_1 \|_{M_2}(\chi(S_1),...,\chi(S_n))\} \tag*{}\\ 
	 &= AX(\|\phi_1\|_{M_2}(\chi(S_1),...,\chi(S_n)))\tag*{}\\
	 &= AX \circ \| \phi_1 \|_{M_2}(\chi(S_1),...,\chi(S_n))\tag*{}\\
	 &= \| AX\phi_1 \|_{M_2}(\chi(S_1),...,\chi(S_n)).\tag*{}
	\end{align}	
	  The proof is similar for $\phi = EX\phi_1$.
	
	\item For $\phi = K_a \phi_1$, $\|\phi \| = K_a \circ \| \phi_1 \|$. Then, 
	\begin{align*}
	\chi(\| \phi \|_{M_1}&(S_1,...,S_n)) = \chi(K_a(\|\phi_1\|_{M_1}(S_1,...,S_n)))  \\
	&= \chi(\{ x \in \supp(t_{M_1}) \mid \forall y \in \supp(t_{M_1}) \text{ with } x \sim_a y \\
		& \qquad \qquad \text{ we have } y \in \|\phi_1 \|_{M_1}(S_1,...,S_n)\})  \\
	&= \{\chi(x) \in \supp(t_{M_2}) \mid s \in \supp(t_{M_1} \text{ and } \forall y \in \supp(t_{M_1} \text{ with } x \sim_a y \\
		& \qquad \qquad \text{ we have } y \in \|\phi_1\|_{M_1}(S_1,...,S_n)\} \\
	&= \{ x' \in \supp(t_{M_2}) \mid \forall \chi^{-1}(y') \in \supp(t_{M_1} \text{ with }\chi^{-1}(x') \sim_a \chi^{-1}(y') \\
		& \qquad \qquad \text{ we have } \chi^{-1}(y') \in \|\phi_1\|_{M_1}(S_1,...,S_n)\} \\
	&=  \{x' \in \supp(t_{M_2}) \mid \forall y' \in \supp(t_{M_2}) \text{ with } x' \sim_a y' \\
		& \qquad \qquad \text{ we have } y' \in \chi(\|\phi_1\|_{M_1}(S_1,...,S_n))\} \\
	&= K_a(\chi(\|\phi_1\|_{M_1}(S_1,...,S_n))) \\
	&= K_a(\|\phi_1\|_{M_2}(\chi(S_1),...,\chi(S_n))) \\
	&= K_a \circ \|\phi_1 \|_{M_2}(\chi(S_1),...,\chi(S_n))\\
	&= \| \phi \|_{M_2}(\chi(S_1),...,\chi(S_n)).	
	\end{align*}
	A similar proof can be given for $\phi = P_a \phi_1$.

	\item For $\phi = \mu Z_i.\phi_1$, $\| \phi \| = \lfp^i_{\| \phi\| }$. Hence,
	\begin{align}
	\chi &(\| \phi \|_{M_1}(S_1,...,S_n)) = \chi( \lfp^i_{\| \phi\|_{M_1} }(S_1,...,S_n)) \tag*{} \\
	&= \chi(\lfp_{\| \phi_1\|_{i,M_1}}(S_1,...,S_{i-1},\cdot,S_{i+1},...,S_n)) \tag*{}\\
	& = \chi( \min \{S \mid \|\phi_1\|_{i,M_1}(S_1,...,S_{i-1},S,S_{i+1},...,S_n) = S\}) \tag*{} \\ 
	& = \big\{ \min \chi(S) \mid S \text{ s.t. } \|\phi_1\|_{i,M_1}(S_1,...,S_{i-1},S,S_{i+1},...,S_n) = S \big \} \tag*{$\chi$ is monotonous} \\
	&= \min \big \{ \chi(S) \mid \chi(\|\phi_1\|_{i,M_1}(S_1,...,S_{i-1},S,S_{i+1},...,S_n)) = \chi(S) \big \} \tag*{}\\
	&= \lfp_{\chi(\|\phi_1\|_{i,M_1}(S_1,...,S_{i-1},\cdot,S_{i+1},...,S_n))} \tag*{inductive step}\\
	&= \lfp_{\|\phi_1\|_{i,M_2}(\chi(S_1),...,\chi(S_{i-1}),\cdot,\chi(S_{i+1}),...,\chi(S_n))} \tag*{}\\
	&= \lfp^i_{\|\phi_1\|_{M_2}}(\chi(S_1),...,\chi(S_n)) \tag*{}\\
	& = \|\phi\|_{M_2}(\chi(S_1),...,\chi(S_n)). \tag*{}
	\end{align}
	The proof is similar for $\phi = \nu Z_i.\phi_1$.
\end{enumerate}
\end{proof}

%% file: ATL_MuCalc.tex
Given an $ATL_{iR}$ formula $\phi = \llangle a \rrangle p_1 \UUU p_2$ where $p_1$ and $p_2$ are atomic proposition, a MAS $M$ 
and a finite run $\rho$ in $M$, 
the instance of the model-checking problem $M,\rho \models \phi$
can be translated to an instance of the model-checking problem in the modal $\mu$-calculus of knowledge of the following formula:
 \vspace{-1.5mm}
$$
\mu Z. \bigvee_{\alpha \in Act_a} K_a \Big( p_2 \vee past_{p_2} \vee  \big( p_1 \wedge 
\!\!\!\bigwedge_{\beta \in Act_{Ag\setminus \{a\}}} \!\!\! [\alpha,\beta] Z \big)\Big) 
 \vspace{-2mm}
$$
and the \emph{modified} system $M'$, in which are created some copies of the successors of the states $s$ labelled with the atomic proposition $p_2$ and the corresponding paths. The copies are labelled with the existing atomic propositions in the successor of $s$ to which is added the new atomic proposition 
$past_{p_2}$. It will label all the states occurring \emph{after} state $s$ carrying a $p_2$.
This mechanism is similar with the ``bookkeeping'' 
employed in the two-player games utilized in \cite{dima-gandalf} for 
checking whether the same formula $\phi$ holds at a state of a MAS.

The formalisation of the modification of $M$ is given below:
For any multi-agent system $M= (Q,Ag,\delta,q_0,\Pi,(\Pi_a)_{a\in Ag},\pi, (Act_a)_{a\in Ag})$,  we compute the multi-agent system $M'$ such that $M'=(Q',Ag,\delta',q'_0,\Pi,(\Pi_a)_{a\in Ag},\pi', (Act'_a)_{a\in Ag})$ with $Q' = Q \times \{0,1\}$, $q'_0 = (q_0,0)$, $\pi'(q,0) = \pi(q)$, $\pi'(q,1) = \pi(q) \cup \{ past_{p_2} \}$, $Act'_a = Act_a \times \{ 0,1 \}, \forall a \in Ag$ and the transition relation defined as: 
 For any transition $q \xrightarrow{(\alpha, \beta)} r$ where $\alpha \in Act_a$ and $\beta = \beta_1, \cdots, \beta_n$ with $\beta_i \in Act_{Ag \setminus \{a\}}$, in $M'$ we have:
 \begin{itemize}
 \item $(q,0) \xrightarrow{((\alpha,0), \beta)} (r,0)$
 \item $(q,1) \xrightarrow{((\alpha,x), \beta)} (r,1)$, $x \in \{0,1\}$
 \item $(q,0) \xrightarrow{((\alpha,1), \beta)} (r,1)$ if $p \in \pi(q)$
 \item $(q,0) \xrightarrow{((\alpha,1), \beta)} (r,0)$ if $p \not \in \pi(q)$
 \end{itemize}

Given a run $\rho = q_0 \xrightarrow{\alpha_1} q_1 \xrightarrow{\alpha_2} ...$ we denote $q_i$ by $\rho[i]$, $i=0,...,|\rho|$ and $\alpha_{i+1}$ by $act(\rho,i)$, $i=0,...,|\rho|-1$. 
We redefine the tree unfolding for the multi player games as being a partial mapping $t_M = (t_M^{node}, t_M^{edge})$ with $t_M^{node}: \Nset \sdpart Q$ , $t_M^{node}(x) = \pi(x[|x|])$ and $t_M^{edge}: \Nset \sdpart \Pi_{a \in Ag} Act_a$, $t_M^{edge}(xi) = \alpha$ if $x \xrightarrow{\alpha} xi$ where $\alpha = (\alpha_1,...,\alpha_k) \in \Pi_{a \in Ag} Act_a$.
In this case, we say that two runs $\rho$ and $\rho'$ are indistinguishable (observationally equivalent) to a coalition $A$ (and note $\rho \sim_A \rho'$) if $|\rho|=|\rho'|$, $act(\rho,i)\restr{A} = act(\rho',i)\restr{A}$ for all $i< |\rho|$, and $\pi_A(\rho[i]) = \pi_A(\rho'[i])$ for all $i \leq |\rho|$.

We define a strategy, as it is defined in \cite{dima-gandalf}, $\sigma$ for a coalition $A$ as any mapping $\sigma : \big( 2^{\Pi_A} \big)^* \sd Act_A$. A strategy $\sigma$ is compatible with a run $\rho = q_0 \xrightarrow{\alpha_1} q_1 \xrightarrow{\alpha_2} ...$ if $\sigma(\pi_A(\rho[0])...\pi_A(\rho[i])) = \alpha_{i+1} \restr{A}$ for all $i \leq |\rho|$. If $\sigma$ is compatible with a run $\rho$, then it is compatible with any run that is indistinguishable from $\rho$ to A.

We also use $[\alpha]p$ to express the fact that for all the successors $xi$ of $x$ for which $t_M^{edge}(xi) = \alpha$ we have that $p \in \pi(t_M^{node}(xi))$.

In order to prove the equivalence between the two problems, we prove that 
for any system $M$ and any $ATL_{iR}$ formula $\phi = \llangle a \rrangle p_1 \UUU p_2$, there exists a system $M'$ as defined below and a formula $\phi' = \mu Z. \bigvee_{\alpha \in Act_a} K_a \Big( p_2 \vee past_{p_2} \vee \big( p_1 \wedge \bigwedge_{\beta \in Act_{Ag\setminus \{a\}}} [\alpha,\beta] Z \big)\Big)$ such that for any run $\rho$ in $M$ and any run $\br \rho$ in $M'$ for which the projection in $M$ is $\rho$, $M,\rho \models \phi$ if and only if $M',\br \rho \models \phi'$.

%% file: theorem2.tex
\begin{proof}
We proceed by structural induction
on the formula $\phi$.
%
%
Note first that the diagram \ref{diagram:plain} holds for the base cases:
\begin{align*}
& t_M^{-1}(\lceil p \rceil) =  \|p\| \qquad t_M^{-1}(\lceil \neg p \rceil) =  \|\neg p\| \\
& t_M^{-1}(\lceil Z_i \rceil (S_1, \ldots,S_n))  = \| Z_i \| (S_1,\ldots,S_n) \text{ for all $S_j \subseteq Q$, $1\leq j\leq n$.}
\end{align*}

The induction step relies on two groups of properties: 
on one side, the commutativity  of $t_M^{-1}$
with finite unions/intersections, and two commutativity
diagrams relating $t_M^{-1}$ with the mappings $AX$/$AX^f$, resp. $EX$/$EX^f$.
The second group of properties
is represented by two characterizations for 
the restrictions of $\lfp^i_{\|\phi\|}$ and $\gfp^i_{\|\phi\|}$ 
on $t_M$-regular sets of nodes of $t_M$.

The first group of properties is summarized in the following identities:
\begin{enumerate}
\item For any two sets $S_1,S_2 \subseteq Q$,
\begin{align*}
t_M^{-1}(S_1 \cup S_2) & = t_M^{-1}(S_1) \cup t_M^{-1}(S_2)\\
t_M^{-1}(S_1 \cap S_2) & = t_M^{-1}(S_1) \cap t_M^{-1}(S_2)
\end{align*}
\item For any set $S \subseteq Q$,
\begin{align*}
t_M^{-1}(AX^f(S)) & = AX(t_M^{-1}(S)) \\
t_M^{-1}(EX^f(S)) & = EX(t_M^{-1}(S)) 
\end{align*}
\end{enumerate}

The following property is essential for the induction step involving the fixpoint operators:
\begin{claim}\label{claim:fixpoint}
Suppose $\phi$ is a \mucalc formula for which the commutative diagram \ref{diagram:plain} holds.
Given $\br S \in \big(2^Q\big)^n$ with $\br S = (S_1,\ldots,S_n)$ 
and an index $i\leq n$,
denote $\hat\phi^i_{\br S}$ the function 
\begin{align}
\hat\phi^i_{\br S} & : 2^{\supp(t_M)} \sd 2^{\supp(t_M)}\notag \\
\hat\phi^i_{\br S} (T) & = \|\phi\|(t_M^{-1}(S_1),\ldots,t_M^{-1}(S_{i-1}),T,
 t_M^{-1}(S_{i+1}),\ldots,t_M^{-1}(S_n)) 
\end{align}
Also denote 
$\lceil \hat \phi^i_{\br S} \rceil $ the function $\lceil \hat\phi^i_{\br S}\rceil : 2^{Q} \sd 2^{Q}$ with 
\[
\lceil \hat\phi^i_{\br S}\rceil (R) = \lceil \phi\rceil(S_1,\ldots,S_{i-1},R,S_{i+1},\ldots,S_n) 
\]


Then $\lfp_{\hat\phi^i_{\br S}} = 
t_M^{-1} (\lfp_{\lceil \hat\phi^i_{\br S}\rceil})$ and  
$\gfp_{\hat\phi^i_{\br S}} = 
t_M^{-1} (\gfp_{\lceil \hat\phi^i_{\br S}\rceil}))
$.

\end{claim}
\smallskip

\input{proof2.tex}







Returning to the proof of Theorem \ref{teo:plain},
 the induction step concerning the least fixed point follows easily:
\begin{align*}
\|\lfp^i_\phi\| \circ (t_M^{-1})^n (S_1,\ldots,S_n) & = \| \lfp^i_\phi\|(t_M^{-1}(S_1),...,t_M^{-1}(S_n)) = \lfp^i_{\| \phi \|(t_M^{-1}(S_1),...,t_M^{-1}(S_n))} \\
& = \lfp_{\hat\phi^i_{(S_1,\ldots,S_n)}} 
 = t_M^{-1} (\lfp^i_{\lceil \phi\rceil }(S_1,\ldots,S_n))
\end{align*}
where in the last step we utilized the claim above. 
A similar proof gives the commutation property for the greatest fixpoint. \hfill$\Box$
\end{proof}

%% file: proof2.tex
\begin{proof}
We may prove by induction on $j \in \Nset$ that 
\begin{equation}\label{id:hat}
(\hat\phi^i_{\br S})^j(\emptyset) = t_M^{-1} (\lceil \hat\phi^i_{\br S}\rceil^j(\emptyset))
\end{equation}
where the first empty set is an element of $2^{\supp(t_M)}$ whereas the second is an element 
of $2^Q$.

The base case is straightforward, since, for $j=0$, Identity \ref{id:hat} reduces
to $\emptyset = t_M^{-1}(\emptyset)$.

For the induction step we may use the induction hypothesis 
about the commutative diagram \ref{diagram:plain} (applied for producing the third identity below)
to conclude that:
\begin{align*}
(\hat\phi^i_{\br S})^{j+1}&(\emptyset) = 
(\hat\phi^i_{\br S}) ( (\hat\phi^i_{\br S})^j(\emptyset) ) \\
& = (\hat\phi^i_{\br S}) ( t_M^{-1}(\lceil \hat\phi^i_{\br S}\rceil^j(\emptyset) ))\\
& = \|\phi\| (t_M^{-1}(S_1),\ldots,t_M^{-1}(S_{i-1}), t_M^{-1}(\lceil \hat\phi^i_{\br S}\rceil^j(\emptyset) ),
 t_M^{-1}(S_{i+1}),\ldots,t_M^{-1}(S_n))  \\
& = t_M^{-1} ( \lceil \phi\rceil(S_1,\ldots,S_{i-1}, \lceil \hat\phi^i_{\br S}\rceil^j(\emptyset) ,
S_{i+1},\ldots,S_n) ) \\
& = t_M^{-1} ( \lceil \hat\phi^i_{\br S}\rceil  ( \lceil \hat\phi^i_{\br S}\rceil^j (\emptyset) )) \\
& = t_M^{-1} (\lceil \hat\phi^i_{\br S}\rceil^{j+1}(\emptyset))
\end{align*}


We then need to prove that $\Big(\lceil \hat\phi^i_{\br S}\rceil^j(\emptyset)\Big)_{j\geq 0}$
is an increasing sequence of subsets of $Q$. To that end, we will prove that $\lceil \hat \phi^i_{\cdot} \rceil(\cdot)$ is  monotonously increasing in both arguments. That is, $\lceil \hat \phi^i_{\br S_1} \rceil (S') \subseteq \lceil \hat \phi^i_{\br S_2} \rceil (S'')$ 
for all $S' \subseteq S'' \in 2^Q$ and for all $\br S_1 = (S_{11},S_{12},...,S_{1n})$ and $\br S_2 = (S_{21},S_{22},...,S_{2n})$ 
with $S_{1k} \subseteq S_{2k} \text{ for all } 1 \leq k \leq n$.
\input{demMonoton.tex}

 This sequence stabilizes at a certain integer $k$, 
which is the fixpoint of $\lceil \hat\phi^i_{\br S}\rceil $:
\begin{equation}\label{id:fixpoint-Q}
\lfp_{\lceil \hat\phi^i_{\br S}\rceil} =
\lceil \hat\phi^i_{\br S}\rceil^{k} (\emptyset) = 
\lceil \hat\phi^i_{\br S}\rceil^{k+1}(\emptyset)
\end{equation}

As a consequence of this and of Identity \ref{id:hat}, the fixpoint of 
$\hat\phi^i_{\br S}$ is reached for 
\[
\lfp_{\hat\phi^i_{\br S}}  =
t_M^{-1}\big(\lceil \hat\phi^i_{\br S}\rceil^{k} (\emptyset)\big)
\]
which ends the proof of Claim \ref{claim:fixpoint}
\end{proof}

%% file: demMonoton.tex
This can be proved by induction on the structure of $\phi$ as follows:
\begin{enumerate}
	\item For $\phi = p$ or $\phi = \neg p$, the property holds since in this case $\lceil \phi \rceil$ is constant.
	\item For $\phi = Z_r$, $\lceil \phi \rceil$ is the $r$-th projection.
		If $r=i$, then $\lceil \hat \phi^i_{\br S_1} \rceil (S') = S' \subseteq S'' = \lceil \hat \phi^i_{\br S_2} \rceil (S'')$. 
		Otherwise, $\lceil \hat \phi^i_{\br S_1} \rceil (S') = S_{1r} \subseteq S_{2r} = \lceil \hat \phi^i_{\br S_2} \rceil (S'')$.
	
	\item For $\phi = \phi_1 \vee \phi_2$, $\lceil \phi \rceil = \lceil \phi_1 \rceil \cup \lceil \phi_2 \rceil$. 
	Since $\lceil \hat \phi^i_{\br S} \rceil (S') = \lceil \phi \rceil (S_1,...,S_{i-1},S',S_{i+1},...,S_n)$, 
	assuming that the property holds for $\phi_1$ and $\phi_2$, 
	$\lceil \hat \phi^i_{1,\br S_1} \rceil (S') \subseteq \lceil \hat \phi^i_{1,\br S_2} \rceil (S'')$ 
	and $\lceil \hat \phi^i_{2,\br S_1} \rceil (S') \subseteq \lceil \hat \phi^i_{2,\br S_2} \rceil (S'')$. 
	Hence, $\lceil \hat \phi^i_{1,\br S_1} \rceil (S') \cup \lceil \hat \phi^i_{2,\br S_1} \rceil (S') \subseteq \lceil \hat \phi^i_{1,\br S_2} \rceil (S'') \cup \lceil \hat \phi^i_{2,\br S_2} \rceil (S'')$ 
	and therefore $\lceil \hat \phi^i_{\br S_1} \rceil (S') \subseteq \lceil \hat \phi^i_{\br S_2} \rceil (S'')$. \\
	For $\phi = \phi_1 \wedge \phi_2$ the proof is similar.
	
	\item For $\phi = AX \phi_1$, $\lceil \phi \rceil = AX_f \circ \lceil \phi_1 \rceil = \{ q \in Q \mid \forall r \in Q, \text{ if } q \rightarrow r \in \delta \text{ then } r \in \lceil \phi_1 \rceil \}$.
	 That is, the predecessors of nodes in $\lceil \phi_1 \rceil$ that have no successors outside $\lceil \phi_1 \rceil$.\\
	It is easy to see that $AX_f$ is monotonous using the definition: given $S_1, S_2 \in 2^Q$, $S_1 \subseteq S_2$, we have  
	\begin{align}		
	&AX_f(S_1) = \{ q \in Q \mid \forall r \in Q, \text{ if } q \rightarrow r \in \delta \text{ then } r \in S_1 \} \tag*{}
	&\intertext{ and since $S_1 \subseteq S_2$, if $r \in S_1$, then $r \in S_2$. Hence, }
	&AX_f(S_1) \subseteq \{ q \in Q \mid \forall r \in Q, \text{ if } q \rightarrow r \in \delta \text{ then } r \in S_2 \} = AX_f(S_2). \tag*{}
\end{align}
	Then, since $\lceil \hat \phi^i_{1, \br S_1} \rceil (S') \subseteq \lceil \hat \phi^i_{1, \br S_2} \rceil (S'')$, 
	by applying $AX_f$ which is monotonous, we obtain that $AX_f(\lceil \hat \phi^i_{1, \br S_1} \rceil (S')) \subseteq AX_f(\lceil \hat \phi^i_{1, \br S_2} \rceil (S''))$. 
	That is, $\lceil \hat \phi^i_{\br S_1} \rceil (S') \subseteq \lceil \hat \phi^i_{\br S_2} \rceil (S'')$. 
	A similar proof can be given for $\phi = EX\phi_1$.
	
	\item For $\phi = K_a \phi_1$, $\lceil \phi \rceil(S) = K^{f}_a \circ \lceil \phi_1 \rceil(S) = \{ q \in Q \mid \forall s \in Q \text{ with } (s,q) \in \Gamma_a \text{ then } s \in \lceil \phi_1 \rceil(S) \}$. 
	We can prove, as we did in the case of $AX_f$, that $K^{f}_a$ is monotonous and 
	applying it to $\lceil \hat \phi^i_{1, \br S_1} \rceil (S') \subseteq \lceil \hat \phi^i_{1, \br S_2} \rceil (S'')$ 
	from the inductive hypothesis, we obtain that $K^{f}_a ( \lceil \hat \phi^i_{1, \br S_1} \rceil (S')) \subseteq K^{f}_a ( \lceil \hat \phi^i_{1, \br S_2} \rceil (S''))$. 
	That is,$\lceil \hat \phi^i_{\br S_1} \rceil (S') \subseteq \lceil \hat \phi^i_{\br S_1} \rceil (S')$. \\
	For $\phi = P_a \phi_1$ the proof results from the duality of $K_a^f$ and $P_a^f$, i.e., $P_a^f(S) = \br{K_a^{f}(\br S)}$. 
	We have that $\lceil \hat \phi^i_{1, \br S_1} \rceil (S') \subseteq \lceil \hat \phi^i_{1, \br S_2} \rceil (S'')$ and then $\br{\lceil \hat \phi^i_{1, \br S_1} \rceil (S')} \supseteq \br{\lceil \hat \phi^i_{1, \br S_2} \rceil (S'')}$. Applying $K_a^{f}$ to it and then computing the dual set, we have that $\br{K_a^{f}(\br{\lceil \hat \phi^i_{1, \br S_1} \rceil (S')})} \subseteq \br{K_a^{f}(\br{\lceil \hat \phi^i_{1, \br S_2} \rceil (S'')})}$. That is, $P^{f}_a ( \lceil \hat \phi^i_{1, \br S_1} \rceil (S')) \subseteq P^{f}_a ( \lceil \hat \phi^i_{1, \br S_2} \rceil (S''))$.
	
	\item For $\phi = \mu Z_r.\phi_1$, $\lceil \phi \rceil = \lfp_{\lceil \phi_1 \rceil}^r$. 
	We have that \\
	$\lceil \hat \phi^i_{\br S_1} \rceil (S')$ = $\lceil \phi \rceil (S_{1,1},...,S_{1,i-1},S',S_{1,i+1},...,S_{1,n})$ \\ 
	= $\lfp_{\lceil \phi_1 \rceil_r(S_{1,1},...,S_{1,r-1},\cdot,S_{1,r+1},...,S_{1,i-1},S',S_{1,i+1},...,S_{1,n})}$.\\
	From the inductive hypothesis we have that $\lceil \hat \phi^i_{1,\br S_1} \rceil (S') \subseteq \lceil \hat \phi^i_{1,\br S_2} \rceil (S'') $ and then, by applying $\lfp$, we obtain that
	\begin{align*}	 
	&\lfp_{\lceil \phi_1 \rceil_r(S_{1,1},...,S_{1,r-1},\cdot,S_{1,r+1},...,S_{1,i-1},S',S_{1,i+1},...,S_{1,n})} \\
	& \subseteq \lfp_{\lceil \phi_1 \rceil_r(S_{2,1},...,S_{2,r-1},\cdot,S_{2,r+1},...,S_{2,i-1},S',S_{2,i+1},...,S_{2,n})}
	\end{align*}
	and then 
	 $\lceil \hat \phi^i_{\br S_1} \rceil (S') \subseteq \lceil \hat \phi^i_{\br S_1} \rceil (S')$.
\end{enumerate}
	
	We may conclude that $\lceil \hat \phi^i_{\br S} \rceil$ is monotonously increasing. 
	Now, we know that $\emptyset \subseteq \lceil \hat \phi^i_{\br S} \rceil(\emptyset)$. 
	
	By induction, we can prove that 
	$\lceil \hat \phi^i_{\br S} \rceil^j(\emptyset) 
	\subseteq \lceil \hat \phi^i_{\br S} \rceil^{j+1}(\emptyset), \forall j \geq 0$. 
	Hence, $(\lceil \hat \phi^i_{\br S} \rceil^j(\emptyset))_{j \geq 0}$ is an increasing sequence of subsets.

%% file: mu-nonmix.tex

\section{A Fragment of the $\mu$-calculus of Knowledge with a Decidable Model-Checking Problem}



In this section, we first introduce some additional notations and notions.
Given a MAS $M$ and two agents $a_1,a_2 \in Ag$,
we say that the two agents \textbf{have compatible observability}
if either $\Pi_{a_1} \subseteq \Pi_{a_2}$ or $\Pi_{a_1} \supseteq \Pi_{a_2}$.

Given a formula $\phi$, let $T_\phi$ denote the syntactic tree of $\phi$. We also consider that, in $T_\phi$, each node labeled with a \emph{variable} also has a \emph{successor}, labeled with $\top$.
This convention brings the property that each node in $T_\phi$ whose formula
is a variable has a closed subformula (which is $\top$).

The syntactic tree is constructed by structural induction, with 
\begin{itemize}
\item $\supp(T_p) = \{ \epsilon \}$, $T_p(\epsilon) = p$,
\item $\supp(T_{\neg p}) = \{ \epsilon \}$, $T_{\neg p}(\epsilon) = \neg p$,
\item $\supp(Z) = \{ \epsilon, 1 \}$, $T_Z(\epsilon) = Z$, $T_Z(1)= \top$,

\item $\supp(T_{Op \phi_1}) = \{\epsilon\} \cup \{ 1x \mid x \in \supp(\phi_1)\}$, $T_{Op \phi_1}(\epsilon) = Op$, $T_{Op \phi_1}(1x) = T_{\phi_1}(x)$, where $Op \in \{ AX,EX, K_a, P_a, \mu Z, \nu Z \}$

\item $\supp(T_{\phi_1 Op \phi_2}) = \{\epsilon\} \cup \{ 1x \mid x \in \supp(\phi_1)\} \cup \{ 2x \mid x \in supp(\phi_2) \}$, $T_{\phi_1 Op \phi_2}(\epsilon) = Op$, $T_{\phi_1 Op \phi_2}(1x) = T_{\phi_1}(x)$, $T_{\phi_1 Op \phi_2}(2x) = T_{\phi_2}(x)$, $Op \in \{ \wedge, \vee \}$

\end{itemize}

We then denote $form(x)$ the subformula of $\phi$ whose 
syntactic tree is $T_\phi\restr{x}$, i.e. the subtree of $T_\phi$ rooted at $x$, 
and say that $x$ is \textbf{closed} if $form(x)$ is closed.

We then say that an epistemic operator $Op \in \{K_a,P_a\mid a\in Ag \}$ is \textbf{non-closed} at a node $x$ 
in a formula $\phi$ if $form(x)$ is not closed, $Op$ labels a node $y \succeq x$ 
and for all the nodes $y'$ lying on the path between $x$ and $y$
we have that $form(y')$ is not closed.

For each node $x \in \supp(T_\phi)$, we also define 
$AgNCl_\phi(x)$ as being the set of agents $a$ for which $K_a$ or $P_a$ 
is non-closed at $x$.
In addition, given two distinct nodes $x_1 \prec x_2$ with $x_2$ being closed, 
we say that $x_2$ is a \emph{nearest closed successor} of $x_1$ if 
no other closed node lies on the path from $x_1$ to $x_2$.

\begin{definition}
The \textbf{\munonmix }
is the fragment of the \muknl consisting of formulas $\phi$ 
satisfying the following property:
\begin{quote}
Any two agents $a$ and $b$ for which there exist epistemic operators 
$Op_a \in \{K_a,P_a\}$, $Op_b \in \{K_b,P_b\}$ such that both $Op_a$ and $Op_b$ 
are not closed at some node $x$ of $T_\phi$
must have compatible observability, i.e. $\Pi_a \subseteq \Pi_b$ or $\Pi_b\subseteq \Pi_a$.
\end{quote}
\end{definition}



All formulas of $KB_n$ \cite{HalpernVardi86,HalpernVardi89}, that is, 
$CTL$ with individual knowledge operators, 
are formulas of the \munonmix.
Other examples of nonmixing formulas are the following ($a$ and $b$ are two agents such that $\Pi_a\subseteq\Pi_b$):
\vspace{-1.5mm}
$$\mu Z_1 . (p \vee K_a( EX .Z_1) \wedge \nu Z_2.(q \wedge Z_1 \wedge K_a (EX Z_2)))  $$
$$\mu Z_1 . (p \vee K_a( EX .Z_1) \wedge \nu Z_2.(q \wedge K_b (EX Z_2))) $$

\vspace{-2mm}
\noindent
Examples of formulas that are not in the \munonmix are ($a$ and $b$ are two agents such that $\Pi_a\not\subseteq\Pi_b$ and $\Pi_b\not\subseteq\Pi_a$):
\vspace{-1.5mm}
$$C_{a,b} \phi = \nu Z . (\phi \wedge K_a Z \vee K_b Z ) $$
$$\mu Z_1 . (p \vee K_a( EX .Z_1) \wedge \nu Z_2.(q \wedge Z_1 \wedge K_b (EX Z_2)))$$

\vspace{-4mm}
\begin{theorem}\label{mck-nonmix}
The model-checking problem for the \munonmix is decidable.
\vspace{-2mm}
\end{theorem}

The crux of the proof consists of proving a
commutativity property relating $t_M^{-1}$ with the operators $K_a$/$K_a^f$, resp. $P_a$/$P_a^f$, similar with the 
properties relating  $t_M^{-1}$ with $AX$/$AX^f$, resp. $EX$/$EX^f$.
Unfortunately, this commutativity property does not hold for any 
MAS $M$, as it is shown by the following example. 

\begin{figure}[ht]
\begin{minipage}{6cm}
\vspace*{-6pt}
\begin{center}
\vspace*{-10pt}
\begin{tikzpicture}[->,bend angle=25,auto,node distance=.9cm, scale=.6] 

\tikzstyle{nod}=[scale=.7, circle,thick,draw=black,minimum size=10pt,inner sep=1pt]

  \node (n0) at (5,12) {};
  \node[nod] (n1) at (5,11) {1,$p_1$};
  \node[nod] (n2) at (8,11) {2,$p_1$};
  \node[nod] (n3) at (2,11) {3,$p_1$};
	
  \path (n0) edge (n1);
  \path (n1) edge (n2);
  \path (n2) edge (n1);
  \path (n1) edge (n3);
  \path (n3) edge[loop above] (n3);
  
\end{tikzpicture}

(a)
\end{center}
\end{minipage}
\begin{minipage}{6cm}
\begin{center}
\vspace*{-6pt}
\input{unfolding.tex}
\vspace*{-5pt}

(b)
\end{center}
\end{minipage}
\vspace{-3mm}
\caption{(a) 
A one-agent system with $\Pi_a = \{p_1\}$, 
(b) The unfolding of the system in Fig.~\ref{fig:noncommut1}(a)
}
\label{fig:noncommut1}
\vspace{-5mm}
\end{figure}
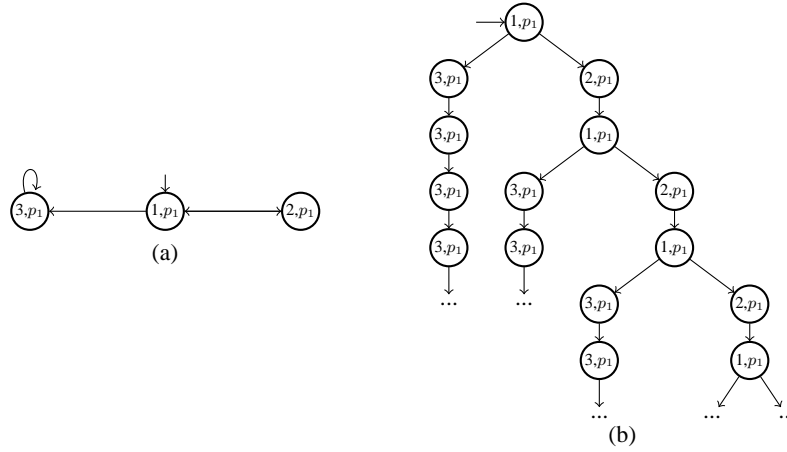

\vspace{-2mm}
\begin{example}
Let $M$ be the one-agent system in Fig.~\ref{fig:noncommut1}. 
If we put $S = \{1,3\}$ then using Figure \ref{fig:noncommut1} b) and the definitions if $K_a/P_a$ and $K_a^f/P_a^f$, we have that
$K_a^f(\{1,3\}) = \br{P_a^f(\{2\})} = \br{\{ 2,3 \}} = \{ 1 \}$.That is, 
 $$t_M^{-1}(K^f_a(S)) = \{ x \in \supp(t) \mid x[|x|] = 1 \}$$
 Similarly,
  \begin{align*}
  K_a(t_M^{-1}(\{ 1,3 \}))& = \br{P_a(t_M^{-1}(\{ 2 \}))}\\
  & = \br{ \{ x \in \supp(t) \mid x[|x|] = 2 \vee (x[|x|] = 3 \wedge |x| \text{ is even}) \} } \\
  & = \{ x \in \supp(t) \mid x[|x|] = 1 \vee (x[|x|] = 3 \wedge |x| \text{ is odd}) \}.
\end{align*}

We can observe that $t_M^{-1}(K^f_a(S))$ 
contains only nodes of $t_M$ labeled with state $1$, whereas
$K_a(t_M^{-1}(S))$ contains more nodes, in particular 
nodes labeled with $3$ occurring on the odd levels of $t_M$.
\end{example}




\begin{definition}\label{def:in_split}
Given two MASs 
$M_i = (Q_i,Ag,\delta_i,q_0^i,\Pi,(\Pi_a)_{a\in Ag},\pi_i)$  ($i=1,2$) over the same set of atomic propositions,
we say that $M_1$ is an \textbf{in-splitting} of $M_2$ 
if there exists a pair of surjective mappings $\chi = (\chi_{st},\chi_{tr})$, 
with $\chi_{st} : Q_1 \sd Q_2$, $\chi_{tr}: \delta_1\sd \delta_2$ 
satisfying the following properties:
\begin{enumerate}
\item For each $q,r \in Q_1$, $(q,r)\in \delta_1$, $\chi_{tr}((q, r)) = (\chi_{st}(q) , \chi_{st}(r)) \in \delta_2$.
\item For each $q \in Q_1$, $\pi_2(\chi_{st}(q)) = \pi_1(q)$.
\item For each $q \in Q_1$, $\outdegree(\chi_{st}(q)) = \outdegree(q)$, where $\outdegree(q)$ is the number of transitions starting in $q$.
\item $\chi_{st}(q_0^1) = q_0^2$.
\end{enumerate}
The in-splitting is an \textbf{isomorphism} whenever $\chi_{st}$ and $\chi_{tr}$ are bijective.
\end{definition}


Further, the pair $\chi = (\chi_{st},\chi_{tr})$ is called an \emph{in-splitting mapping}. 
Also, we may write $\chi:M_1\sd M_2$ to denote the fact that $\chi = (\chi_{st},\chi_{tr})$ is a witness for 
$M_1$ being an in-splitting of $M_2$.

Note that an in-splitting mapping (term borrowed from symbolic dynamics \cite{LindMarcus}) 
represents a surjective functional bisimulation between two transition systems.
The following proposition can be seen as a generalization of this remark 
(proof given in \cite{bozianu-dima-enea-arxiv}):

\begin{proposition}\label{prop:commut-finite}
Consider two MASs 
$M_i = (Q_i,Ag,\delta_i,q_0^i,\Pi,(\Pi_a)_{a\in Ag},\pi_i)$  ($i=1,2$) over the same set of atomic propositions,
connected by an in-splitting mapping $\chi = (\chi_{st},\chi_{tr}) : M_1 \sd M_2$.
Then for any plain \mucalc formula $\phi$ 
the following diagram commutes:
\begin{equation}
\begin{diagram}[tight,height=20pt,width=4em]
\big(2^{Q_1}\big)^n & \rTo^{\lceil \phi \rceil_{M_1}} & 2^{Q_1} \\
\uTo^{(\chi_{st}^{-1})^n} & & \uTo_{\chi_{st}^{-1}} \\
\big(2^{Q_2}\big)^n & \rTo^{\lceil\phi\rceil_{M_2}} & 2^{Q_2} \\
\end{diagram} 
\label{diagram:refinement} 
\end{equation}
\end{proposition}
\input{proofProp14.tex}


%
%
%
%

\begin{rem}\label{rema:commut-finite}
Proposition \ref{prop:commut-finite} does not hold for general \muknl formulas.
To see this, consider the system depicted in Fig.~\ref{fig:in-splitting} (a), which is an
in-splitting of the system from Fig.~\ref{fig:noncommut1} (a),
resulting from splitting state $3$ in two states, denoted $3$ and $4$,
(i.e. $\chi(1) = 1, \chi(2) = 2, \chi(3) = \chi(4) = 3$)
with transitions $(3,4)\in \delta$ and $(4,4) \in \delta$.

\begin{figure}[ht]
\begin{minipage}{6cm}
\vspace*{-6pt}
\begin{center}
\vspace*{-10pt}
\begin{tikzpicture}[->,bend angle=25,auto,node distance=.9cm, scale=.6] 

\tikzstyle{nod}=[scale=.7, circle,thick,draw=black,minimum size=10pt,inner sep=1pt]

  \node (n0) at (5,12) {};
  \node[nod] (n1) at (5,11) {1,$p_1$};
  \node[nod] (n2) at (7,11) {2,$p_1$};
  \node[nod] (n3) at (3,11) {3,$p_1$};
  \node[nod] (n4) at (1,11) {4,$p_1$};
	
  \path (n0) edge (n1);
  \path (n1) edge (n2);
  \path (n2) edge (n1);
  \path (n1) edge (n3);
  \path (n3) edge (n4);
  \path (n4) edge[loop left] (n4);
  
\end{tikzpicture}

(a)
\end{center}
\end{minipage}
\begin{minipage}{6cm}
\begin{center}
\input{in-splitting.tex}

(b)
\end{center}
\end{minipage}
\vspace*{-5pt}
\caption{(a)
 An in-splitting of the system from Fig.~\ref{fig:noncommut1}
; (b)The unfolding of the system from Fig.~\ref{fig:in-splitting}(a)
}
\label{fig:in-splitting}
\vspace{-5mm}
\end{figure}
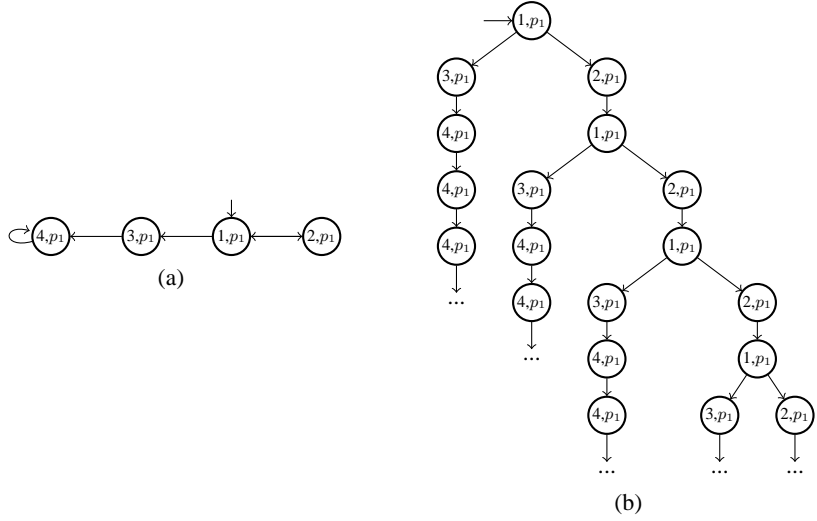


Note that $K_a^f(\{1,4\}) = \br{P_a^f(\{2,3\})} = \br{\{ 2,3 \}} = \{ 1,4 \}$ and hence, 
 $$t_M^{-1}(K^f_a(\{1,4\})) = \{ x \in \supp(t) \mid x[|x|] = 1 \vee x[|x|] = 4 \}$$
 On the other hand,
  \begin{align*}
  K_a(t_M^{-1}(\{ 1,4 \}))& = \br{P_a(t_M^{-1}(\{ 2,3 \}))}\\
  & = \br{ \{ x \in \supp(t) \mid x[|x|] = 2 \vee x[|x|] = 3 \vee (x[|x|] = 4 \wedge |x| \text{ is even}) \} } \\
  & = \{ x \in \supp(t) \mid x[|x|] = 1 \vee (x[|x|] = 4 \wedge |x| \text{ is odd}) \}.
\end{align*}

That is, $t_M^{-1}(K_a^f(S))$ 
contains all nodes of $t_M$ labelled with states $1$ or $4$, whereas
$K_a(t_M^{-1}(S))$ contains fewer nodes, in particular 
nodes labelled with $1$  and nodes labelled with $4$ occurring on the odd levels of $t_M$.





\end{rem}








The following notion corresponds with the ``determinization'' used for model-checking 
LTLK/CTLK \cite{meyden-shilov,dima08clima} or solving 2-player parity games with one player having incomplete information
\cite{chatterjee-doyen-fullpogames}: 

\begin{definition} \label{def:a-distingtion}
Given a MAS $M = (Q,Ag,\delta,q_0,\Pi,(\Pi_a)_{a\in Ag},\pi)$, we define
the multi agent system $\Delta^{pre}_a(M) = (\tilde Q^{pre}, Ag,\tilde \delta,\tilde q_0,\Pi,(\Pi_a)_{a\in Ag},\tilde \pi)$ as follows:
\begin{itemize}
\item $\tilde Q^{pre} = \{(s,S) \mid s\in Q, S \subseteq \{ q \in Q \mid \pi_a(q) = \pi_a(s) \} \}$
and $\tilde q_0 = (q_0, \{q_0\})$.
\item $\tilde \delta$ is composed of all tuples 
of the form $((s,S) , (r,R))$ where $(s, r) \in \delta$ and 
$R = \{r' \in Q \mid \pi_a(r') = \pi_a(r) \text{ and } \exists s' \in S \text{ with } (s' , r') \in \delta\}$.
\item $\tilde \pi(s,S) = \pi(S) = \pi(s)$.
\end{itemize}
The \textbf{$a$-distinction} of $M$, denoted $\Delta_a(M)$, 
is the restriction of $\Delta^{pre}_a(M)$ to reachable states, i.e., 
$\Delta_a(M) = (\tilde Q, Ag,\tilde \delta\mid_{\tilde Q},\tilde q_0,\Pi,
(\Pi_a)_{a\in Ag},\tilde \pi \mid_{\tilde Q})$ where $\tilde Q = \{ \tilde s \in \tilde Q^{pre} \mid \tilde s \text{ is reachable from } \tilde q_0 \}$.
\end{definition}

Given a run $\rho$ in $\Runs(\Delta_a(M))$, we denote $\rho \restr{1}$ the projection of $\rho$ onto its first component.

\begin{lemma}\label{lemma:def_S}
	Given a MAS $M = (Q,Ag,\delta,q_0,\Pi,(\Pi_a)_{a\in Ag},\pi)$, the following two properties hold:
	\begin{enumerate}
	\item For each run $\br \rho$ in $\Delta_a(M)$ ending in $(s,S)$, $$S=\{ r \in Q \mid \exists \rho' \text{ in } M \text{ that ends in } r \text{ with } \rho' \sim_a \br \rho \restr{1} \}$$
	\item For each two runs $\br \rho$, $\br \rho'$ in  $\Delta_a(M)$ with $\br \rho$ ending in $(s,S) \in \tilde Q$, if $ \br \rho' \sim_a \br \rho$, then there exists $r \in Q$ s.t. $\br \rho'$ ends in $(r,S)$.
	\end{enumerate}
\end{lemma}
\input{proof3.tex}

Given a MAS $M = (Q,Ag,\delta,q_0,\Pi,(\Pi_a)_{a\in Ag},\pi)$, and an agent $a \in Ag$, we say that $M$ is \textbf{$a$-distinguished} if $\Gamma_a$ defined on page \pageref{def:Gamma} is a \textbf{congruence relation}, that is, an equivalence relation with the following property:
\vspace{-1mm}
\begin{equation}
 \text{for any } q,r\in Q, \text{ if } q \Gamma_a r, (q , q') \in \delta, (r , r') \in \delta \text{ and } \pi_a(q') = \pi_a(r'), \text{ then } q' \Gamma_a r'. 
 \vspace{-1.5mm}
\end{equation}

\begin{lemma}\label{lemma:Gamma-exists}
	For a MAS $M = (Q,Ag,\delta,q_0,\Pi,(\Pi_a)_{a\in Ag},\pi)$ and an agent $a \in Ag$ with $\Gamma_a$ a congruence relation, we have that $s \Gamma_a r$ if and only if there exists $\rho$ and $\rho' \in \Runs(M)$ s.t. $\rho$ ends in $s$ and $\rho'$ ends in $r$ with $\rho \sim_a \rho'$.
\end{lemma}
\input{proof4.tex}

\begin{lemma}\label{lemma:S=R}
	Given a MAS $M$ and an agent $a \in Ag$, for any two reachable states $(q,S)$ and $(r,R)$ in $\Delta_a(M)$, $(q,S) \Gamma_a (r,R)$ if and only if $S=R$.
\end{lemma}
\input{proof5.tex}



\begin{proposition}\label{prop:distinction-in-splitting}
\begin{enumerate}
\item For any MAS $M$, $\Delta_a(M)$ is an in-splitting of $M$.
We denote this in-splitting as $\Delta^{-1}_{a,M} : \Delta_a(M) \sd M$. 
Whenever the MAS $M$ is clear from the context, we use the notation $\Delta^{-1}_a$ instead of $\Delta^{-1}_{a,M}$.
\item For any agent $a \in Ag$ we have that
$\Delta_a(M)$ is $a$-distinguished.
\end{enumerate}
\end{proposition}
\input{proof6.tex}

\begin{proposition}\label{prop:a-b-distinction}
	For any MAS $M$ and two agents $a,b \in Ag$ with $\Pi_a \subseteq \Pi_b$, if $M$ is $b$-distinguished, then $\Delta_a(M)$ is $b$-distinguished too.
\end{proposition}
\input{proof7.tex}

The close relationship between 
the relation $\Gamma_a$ and the epistemic operators is resumed by the following proposition:
\begin{proposition}\label{prop:restricted}
For any MAS $M$, the following diagram commutes iff $M$ is $a$-distinguished:
\vspace{-3mm}
\begin{equation}
\begin{diagram}[tight,height=20pt,width=4em]
2^Q & \rTo^{K_a^f} & 2^Q \\
\dTo^{t_M^{-1}} & & \dTo_{{t_M}^{-1}}  \\
2^{\supp(t_M)} & \rTo^{K_a} & 2^{\supp(t_M)} \\
\end{diagram} 
\label{diagram:deterministic} 
\end{equation} 

The same holds if the pair $K_a/K_a^f$ is replaced with 
$P_a/P_a^f$.
\end{proposition}

\input{proofProp18.tex}

\begin{definition}
We say that the pair of epistemic operators $K_a/K_a^f$, resp. $P_a/P_a^f$, \textbf{commutes for} 
$M$ if the 
diagram \ref{diagram:deterministic} is commutative for the respective pair. 
\end{definition}
Proposition \ref{prop:restricted} gives the first 
restricted form which may lead to the commutativity of 
Diagram \ref{diagram:plain} for formulas of the \muknl.
The second restricted form in which the pair $K_a/K_a^f$ (resp. $P_a/P_a^f$) commutes with a system 
is stated as point 2 in the next proposition:

\begin{proposition}\label{prop:second-commute}
Consider two MASs $M_i = (Q_i,Ag,\delta_i,q_0^i,\Pi,(\Pi_a)_{a\in Ag},\pi_i)$  ($i=1,2$), 
with $Q_1 = \{1,\ldots, n_1\}$ and $Q_2 = \{1,\ldots, n_2\}$,
related by an in-splitting $\chi = (\chi_{st},\chi_{tr}) : M_1 \sd M_2$,
and define the tree mapping
$\hat \chi : \supp(t_{M_1}) \sd \supp(t_{M_2})$,
where $\hat \chi (\eps)  = \eps$ and $\hat \chi (xi) =  \hat \chi (x) \cdot \chi_{st}(i)$, for any $x\in \supp(t_{M_1})$ and $i\in Q_1$. 
Then the following properties hold:
\begin{enumerate}
\item $\hat \chi$ is a tree isomorphism between $t_{M_1}$ and $t_{M_2}$ and 
$t_{M_2} \circ \hat \chi = \hat\chi \circ t_{M_1}$.
\item For any closed formula $\phi$ of the \muknl
for which the diagram \ref{diagram:plain} commutes in the system $M_2$, 
the following property holds:
\vspace*{-5pt}
$$
\|\phi\|_{M_1} = t_{M_1}^{-1} (\chi_{st}^{-1} ( \lceil \phi \rceil_{M_2} )) 
$$\label{id:commute-muknl}
\vspace*{-5pt}
\end{enumerate}
\end{proposition}
\input{proof8.tex}
\vspace*{-20pt}

\begin{rem}\label{rema:commut-closed}
The previous proposition tells us that, for \emph{closed} formulas of the 
\muknl for which Diagram \ref{diagram:plain} commutes in $M_2$,
in the eventuality that the system $M_2$ needs to be replaced with 
a ``larger'' system $M_1$ (for reasons related with the ``determinization''
that ensures the first type of commutativity of $K_a/P_a$),
the validity of $\phi$ on the tree $t_{M_1}$ 
can be recovered from the set of states 
$\chi_{st}^{-1} \big(\lceil \phi \rceil_{M_2} \big)$, through the 
inverse tree mapping $t_{M_1}^{-1}$.
\end{rem}

We have now the essential ingredients that ensure the
decidability of the model-checking problem for the \munonmix.
The algorithm runs as follows: 
we proceed by constructing the state-transformer interpretations of the 
subformulas of $\phi$ on the given system $M$, 
in a bottom-up traversal of the syntactic tree $T_\phi$.
As long as we only treat subformulas not containing any epistemic operator,
Theorem \ref{teo:plain}
ensures that these state transformers are correct finitary abstractions 
of the tree semantics of our subformulas.

The first time we encounter in $T_\phi$ an epistemic operator $K_a/P_a$, say,
the subformula in the current node is $K_a \phi'$,
we need to replace $M$ with its \emph{$a$-distinction}, $\Delta_a(M)$, 
in order for the appropriate diagram to commute.
This replacement is easier when $\phi'$ is a closed plain \mucalc formula. 
By combining Propositions 
\ref{prop:second-commute} and \ref{prop:restricted},
the tree semantics of the formula $K_a \phi'$ can be computed using 
the state transformer
$K_a^f \big( \Delta_{a}^{-1} \big(\lceil \phi' \rceil_{M} \big) \big) $ in $\Delta_a(M)$,
where $\Delta_{a}^{-1} \big(\lceil \phi' \rceil_{M} \big)$ represents 
the set of states in $\Delta_a(M)$ on which $\phi'$ holds.


The procedure is different when $\phi'$ is non-closed.
In this situation, we cannot determinize $M$, as observed in the remark \ref{rema:commut-finite}.
Therefore we need to descend along the syntactic tree to \emph{all} the ``nearest''
nodes whose formulas are closed, and 
only there apply the $a$-distinction construction, thanks to Proposition \ref{prop:second-commute}.

Suppose even further that $\phi'$ itself contains other knowledge operators,
and some other knowledge operator $K_b$ is encountered
during this descent. The ``nonmixing'' assumption on our formula 
implies that this other agent $b$ has compatible observability with our $a$
($K_a$ and $K_b$ are non-closed at the node associated with $K_a$).
Therefore, the $a$-distinction of the models applied at lower levels commutes with 
$K_b$, fact which is ensured by Proposition \ref{prop:restricted} when the two agents 
have compatible observability.

This whole process ends when we arrive in the root of the syntactic tree,
with an in-splitting $M'$ of the initial system $M$ and a 
(constant) state-transformer $\sigma$, which gives the 
finitary abstraction of the set of nodes of the 
tree $t_M$ where $\phi$ holds.
The following paragraphs formalize this process.

\begin{proof}[Proof of Theorem \ref{mck-nonmix}]
Given a formula $\phi$ in the \munonmix and a MAS $M$, we associate with each node $x$ of $T_\phi$ 
an in-splitting mapping, denoted $T^{Ins}_{\phi}(x)$, 
such that the following properties hold:
\begin{enumerate}

\item For the root $\epsilon$ and any not closed node $x$ in $\supp(T_{\phi})$, $T^{Ins}_{\phi}(\epsilon) = id_M$ and 
 $T^{Ins}_{\phi}(x) = id_{M'}$, with $M$ and $M'$ appropriate MASs.
\item For any $x, xi \in \supp( T_{\phi}), i\in \{1,2\}$, $codom (T^{Ins}_{\phi}(x)) = dom (T^{Ins}_{\phi}(xi))$,
\item For any nodes $x_1, x_2 \in \supp(T_{\phi})$ with $x_1 \preceq x_2$, 
the \textit{in-splitting mapping} between the two nodes is the composition of the mappings from $x_1$ to $x_2$. Formally, 
$$T^{Ins}_{\phi}(x_1...x_2) =^{not} T^{Ins}_{\phi}(x_1) \circ ... \circ T^{Ins}_{\phi}(x_2)$$
Then, for any $x_1, x_2$ leaves in $T_{\phi}$, $T^{Ins}_{\phi}(\epsilon... x_1) = T^{Ins}_{\phi}(\epsilon ... x_2)$, where $\epsilon$ is the root.

\item For any node $x_1$ which is a nearest closed successor of the root $\epsilon$,
if $AgNCl_\phi(\epsilon) = \{a_1,\ldots,a_k\}$ 
and $\Pi_{a_1}\subseteq \ldots \subseteq \Pi_{a_k}$, then 
$T^{Ins}_{\phi}(x_1)$ has the form:
\vspace*{-6pt}
$$
T^{Ins}_{\phi}(x_1) = \Delta_{a_1}^{-1} \circ \ldots \circ \Delta_{a_k}^{-1} \circ \chi,\mbox{ for some $\chi$,}
$$
\end{enumerate}

Next, assuming that $T^{InS}_{\phi}$ is constructed with all the properties above, we denote 
$InS(T^{Ins}_{\phi}) = T^{Ins}_{\phi}(\epsilon...x)$ with $x$ any leave in $T_\phi$.




The construction of $T^{Ins}_{\phi}$ proceeds by structural induction on $\phi$. Whenever we want to emphasize a property of the root of the syntactic tree $T_{\phi}$, we denote it $\epsilon^\phi$.

For the base case we put $T^{Ins}_{p}(\epsilon) = T^{Ins}_{\neg p}(\epsilon) = id_M$, 
for any $p\in \Pi$.
For $\phi = Z$, $Z \in \ZZZ$, note that, by construction, the root of $T_Z$ has a leaf successor which is the only child node. Then, $T^{Ins}_{Z}(\epsilon) = T^{Ins}_{Z}(1) = id_M$.

For the induction case, 
take a formula $\phi = Op. \phi'$ where $Op \in \{AX, EX, \mu Z, \nu Z\}$,
and assume 
$T^{Ins}_{\phi'}(x)$ is defined. 
Then we put  $T^{Ins}_{\phi}(1x) = T^{Ins}_{\phi'}(x)$ for any node $x$ of $\supp(T_{\phi'})$,
and $T^{Ins}_{\phi}(\epsilon^\phi) = id_{M'}$, where $M' = dom(T^{Ins}_{\phi'} (\epsilon^{\phi'}))$.

Suppose $\phi = K_a \phi'$ or $\phi = P_a \phi'$.
Note that for each node $1x$ which is not closed in $T_\phi$, the node $x$ is not closed in $T_{\phi'}$ either. 
Then we put 
$T^{Ins}_{\phi}(1x) = T^{Ins}_{\phi'}(x) = id_{M'}$, with $M'$ the appropriate MAS. 
We also put $T^{Ins}_{\phi}(\epsilon^\phi) = Id_{M_0}$ for the appropriate $M_0$.
Furthermore, for each closed node $1x_1 \in \supp(T_{\phi})$ which \emph{not} a nearest closed successor of $\epsilon^\phi$,
we put $T^{Ins}_{\phi}(1 x_1) = T^{Ins}_{\phi'}(x_1)$.

Take further a node $1x_1$ which is a nearest closed successor of the root $\epsilon^\phi$ and $AgNCL(\epsilon^\phi) = \{a_1,...,a_k\}$.
By the above property 4 from the induction hypothesis, the in-splitting mapping in $x_1$ is
of the form 
$T^{Ins}_{\phi'}(x_1) = \Delta_{a_1}^{-1} \circ \ldots \circ \Delta_{a_k}^{-1} \circ \chi$
with 
$\Pi_{a_1}\subseteq \ldots \subseteq \Pi_{a_k}$.
On the other hand, by the assumption that $\phi$ is a nonmixing formula,
$a$ must have compatible observability with all the agents $a_1,\ldots, a_k$.
Therefore, there must exist some $i \leq k$ such that 
$\Pi_{a_1}\subseteq \ldots \subseteq \Pi_{a_i} \subseteq \Pi_a \subseteq \Pi_{a_{i+1}} \subseteq \ldots \subseteq \Pi_{a_k}$.
We then define 
\vspace*{-2pt}
$$
T^{Ins}_{\phi}(1x_1) = \Delta_{a_1}^{-1} \circ \ldots \circ \Delta_{a_i}^{-1} \circ \Delta_{a}^{-1} \circ  
\Delta_{a_{i+1}}^{-1} \circ \ldots \circ \Delta_{a_k}^{-1} \circ \chi
$$
Note that the domain and the codomain of each $\Delta^{-1}_{a_j}$, 
($j \leq i$) are different in $T^{Ins}_{\phi}$ from those in $T^{Ins}_{\phi'}$, 
due to the insertion of $\Delta_a^{-1}$.

According to the above constructions for $\phi = K_a \phi'$ of $\phi = P_a \phi'$, 
all the four properties are satisfied by $T^{Ins}_{\phi}$, the fourth one resulting from the construction 
of the in-splitting mapping for the nearest closed successors of the root.

Finally, take $\phi\! = \!\phi_1 Op \, \phi_2$ ($Op \!\in \!\{\wedge,\vee\}$).
If $T^{Ins}_{\phi_1} \!= \! T^{Ins}_{\phi_2}$, put
$T^{Ins}_{\phi}(1n) \!=\! T^{Ins}_{\phi_1}(x)$ for all nodes $x \in \supp(T_{\phi_1})$, 
$T^{Ins}_{\phi}(2n) \!=\! T^{Ins}_{\phi_2}(x)$ for all $n \in \supp(T_{\phi_2})$ and 
$T^{Ins}_{\phi}(\epsilon) = Id_M$.

Suppose now $T^{Ins}_{\phi_1} \neq T^{Ins}_{\phi_2}$. 
Consider
$AgNCl(1) = \{a_1,\ldots,a_k\}$ and  
$AgNCl(2) = \{b_1,\ldots,b_l\}$ with 
$\Pi_{a_1}\subseteq \ldots \subseteq \Pi_{a_k}$ and
$\Pi_{b_1}\subseteq \ldots \subseteq \Pi_{b_l}$.
Take then a node $x_1$ which is a nearest closed successor of the root of $T_{\phi_1}$, $\epsilon^{\phi_1}$,
and a node $x_2$ which is a nearest closed successor of 
$\epsilon^{\phi_2}$.
By the induction hypothesis we have:
\vspace*{-2pt}
\begin{align*}
T^{Ins}_{\phi}(x_1) & = \Delta_{a_1}^{-1} \circ \ldots \circ \Delta_{a_k}^{-1} \circ \chi_1 & InS(T^{Ins}_{\phi_1}) & = T^{Ins}_{\phi}(x_1) \circ \chi_1' \\
T^{Ins}_{\phi}(x_2) & = \Delta_{b_1}^{-1} \circ \ldots \circ \Delta_{b_l}^{-1} \circ \chi_2 & InS(T^{Ins}_{\phi_2}) & = T^{Ins}_{\phi}(x_2) \circ \chi_2' 
\end{align*}
with appropriate in-splittings $\chi_1,\chi_1',\chi_2,\chi_2'$.

On the other hand, by the assumption on $\phi$ being nonmixing, 
for any $i\leq k,j\leq l$, the two agents $a_i$ and $b_j$ 
must have compatible observability.
It therefore follows that 
there exists a reordering of the union $\{a_1,\ldots,a_k\} \cup \{b_1,\ldots,b_l\}$ 
as $\{c_1,\ldots,c_m\}$ such that 
$\Pi_{c_i}\subseteq \Pi_{c_{i+1}}$ for all $i\leq m-1$.
Denote then:
\vspace*{-2pt}
$$
\chi_0 = \Delta_{c_1}^{-1} \circ \ldots \circ \Delta_{c_m}^{-1} 
$$
By Proposition \ref{prop:a-b-distinction}, 
$\chi_0$ is a $c$-distinction for any 
$c \in \{a_1,\ldots,a_k\} \cup \{b_1,\ldots,b_l\}$.
Also, by property 2 of the induction hypothesis, 
$\chi_0$ is independent of the choice of the nodes $x_1,x_2$.

The same property from the induction hypothesis also ensures that, 
for any nearest closed successor $\br x_2$ of $\eps^{\phi_2}$, there exist 
in-splittings $\br\chi^{\phi_2,\br x_2}_2,\tilde\chi^{\phi_2,\br x_2}_2$ such that:
\vspace*{-2pt}
\begin{equation}
T^{Ins}_{\phi}(\br x_2) = \Delta_{b_1}^{-1} \circ \ldots \circ \Delta_{b_l}^{-1} \circ \br \chi^{\phi_2,\br x_2}_2 \qquad
\label{id2:ins_n_2_prim}
InS(T^{Ins}_{\phi_2}) = T^{Ins}_{\phi}(\br x_2) \circ \tilde\chi^{\phi_2,\br x_2}_2
\end{equation}

\vspace*{-2pt}
We will then construct $T^{Ins}_{\phi}(\cdot)$ as follows:
\begin{enumerate}
\item For each closed node $x$ which is a leaf in $T_{\phi_1}$ but not a nearest closed successor of $\epsilon^{\phi_1}$,
we put $T^{Ins}_{\phi}(1x) = T^{Ins}_{\phi_1}(x) \circ \chi_2 \circ \chi_2' $.
\item For each non-leaf, closed node $x$ in $T_{\phi_1}$ 
which is not a nearest closed successor of $\epsilon^{\phi_1}$ we copy
$T^{Ins}_{\phi}(1x) = T^{Ins}_{\phi_1}(x)$.
\item For each nearest closed successor $x$ of $\epsilon^{\phi_1}$ which is not a leaf in $T_{\phi_1}$ 
we put $T^{Ins}_{\phi}(1x) = \chi_0 \circ \chi_1$.
\item For each closed node $x$ which is a leaf in $T_{\phi_1}$ and a nearest closed successor of $\epsilon^{\phi_1}$,
we put $T^{Ins}_{\phi}(1x) = \chi_0 \circ \chi_1 \circ \chi_1' \circ \chi_2 \circ \chi_2' $.
\item For each closed node $x$ which is not a close 
successor of $\epsilon^{\phi_2}$ we copy
$T^{Ins}_{\phi}(2x) = T^{Ins}_{\phi_2}(x)$.
\item For each closed node $x$ which is a nearest closed successor 
of $\epsilon^{\phi_2}$ we put
$T^{Ins}_{\phi}(2x) = \chi_0 \circ \chi_1 \circ \chi_1'\circ \br\chi^{\phi_2,x}_2$,
where $\br\chi^{\phi_2,x}_2$ is the in-splitting mapping 
associated to node $x$ as in Identity \ref{id2:ins_n_2_prim} above.
\item For the root $\epsilon$ and the non-closed nodes $x$ of $T_\phi$, $T^{Ins}_{\phi}(\epsilon) = Id_{M'}$ 
and $T^{Ins}_{\phi}(x) = Id_{M''}$, with $M'$ and $M''$ appropriate MASs. 
\end{enumerate}

It's not difficult to see that the resulting mapping $T^{Ins}_{\phi_2}(\cdot)$ satisfies 
the five desired properties.
More specifically, 
property 2 amounts to the following identity:
\vspace*{-2pt}
$$
InS(T^{Ins}_{\phi}) = \chi_0 \circ \chi_1 \circ \chi_1' \circ \chi_2 \circ \chi_2' 
$$


\input{modelchecking.tex}

\end{proof}

%% file: unfolding.tex
\begin{tikzpicture}[->,bend angle=25,auto,node distance=.9cm, scale=.5] 

\tikzstyle{nod}=[scale=.7, circle,thick,draw=black,minimum size=10pt,inner sep=1pt]

  \node (n0) at (3.5,11) {};
  \node[nod] (n1) at (5,11) {1,$p_1$};
	
  \path (n0) edge (n1);

  \node[nod] (n2) at (3,9.5)  {3,$p_1$};
  \node[nod] (nn11) at (3,8) {3,$p_1$};
  \node[nod] (nn12) at (3,6.5) {3,$p_1$};
  \node[nod] (nn13) at (3,5) {3,$p_1$};
  \node (nn14) at (3,3.5) {...};
  
  \path (n1) edge (n2); 
  \path (n2) edge (nn11);
  \path (nn11) edge (nn12);
  \path (nn12) edge (nn13);	
  \path (nn13) edge (nn14);
  
  \node[nod] (n3) at (7,9.5)  {2,$p_1$};
  \node[nod] (n4) at (7,8)  {1,$p_1$};
   
  \path (n1) edge (n3);
  \path (n3) edge (n4);
  
  \node[nod] (nn21) at (5,6.5) {3,$p_1$};
  \node[nod] (nn22) at (5,5) {3,$p_1$};
  \node (nn23) at (5,3.5) {...};
  
  \path (n4) edge (nn21);
  \path (nn21) edge (nn22);
  \path (nn22) edge (nn23);
  
  \node[nod] (n5) at (9,6.5) {2,$p_1$};  
  \node[nod] (n6) at (9,5) {1,$p_1$};
  
  \path (n4) edge (n5);
  \path (n5) edge (n6);
  
  \node[nod] (nn31) at (7,3.5) {3,$p_1$};
  \node[nod] (nn32) at (7,2) {3,$p_1$};
  \node (nn33) at (7,0.5) {...};
  
  \path (n6) edge (nn31);
  \path (nn31) edge (nn32);
  \path (nn32) edge (nn33);
  
  \node[nod] (n7) at (11,3.5) {2,$p_1$};
  \node[nod] (n8) at (11,2) {1,$p_1$};
  \node (n9) at (10,0.5) {...};
  \node (n10) at (12,0.5) {...};
  
  \path (n6) edge (n7);
  \path (n7) edge (n8);
  \path (n8) edge (n9);
  \path (n8) edge (n10);


\end{tikzpicture}

%% file: proofProp14.tex
\begin{proof}
Let $S_1,...,S_n \subseteq 2^{Q_2}$. We prove hence by structural induction on the structure of the formula $\phi$ that \[ \lceil \phi \rceil_{M_1}(\chi_{st}^{-1}(S_1,...,S_n)) = \chi_{st}^{-1}(\lceil \phi \rceil_{M_2}(S_1,...,S_n)) \]
\begin{enumerate}
	\item For $\phi = p$ we have
	\begin{align*}
	 &\lceil \phi \rceil_{M_1}(\chi_{st}^{-1}(S_1,...,S_n)) = \{ q \in Q_1 \mid p \in \pi_1(q) \} \\
	 &\chi_{st}^{-1}( \lceil \phi \rceil_{M_2}(S_1,...,S_n)) = \chi_{st}^{-1}(\{ q \in Q_2 \mid p \in \pi_2(q) \})
	 \intertext{Since $\pi_2(\chi_{st}(q)) = \pi_1(q)$ and $\chi_{st}$ is surjective, we can conclude that}
	 &\lceil \phi \rceil_{M_1}(\chi_{st}^{-1}(S_1,...,S_n)) = \chi_{st}^{-1}( \lceil \phi \rceil_{M_2}(S_1,...,S_n)) 
	 \end{align*}	
	The proof is similar for $\phi = \neg p$.
	
	\item For $\phi = Z_i$, $\lceil \phi \rceil$ is the $i$-th projection and hence
	\begin{align*}		
	 \lceil \phi \rceil_{M_1} (\chi_{st}^{-1} &(S_1,...,S_n)) = \chi_{st}^{-1} (S_i)
	  = \chi_{st}^{-1}(\lceil \phi \rceil_{M_2}(S_1,...,S_n)).
	\end{align*}
	
	\item For $\phi = \phi_1 \vee \phi_2$, $\lceil \phi \rceil = \lceil \phi_1 \rceil \cup \lceil \phi_2 \rceil$. Then,
	\begin{align}
	\lceil \phi &\rceil_{M_1}(\chi_{st}^{-1}(S_1,...,S_n)) 
	= \lceil \phi_1 \rceil_{M_1}(\chi_{st}^{-1}(S_1,...,S_n) \cup \lceil \phi_2 \rceil_{M_1}(\chi_{st}^{-1}(S_1,...,S_n) \tag*{}\\
	&= \chi_{st}^{-1}(\lceil \phi_1 \rceil_{M_2}(S_1,...,S_n)) \cup \chi_{st}^{-1}(\lceil \phi_2 \rceil_{M_2}(S_1,...,S_n)) \tag*{by induction}\\
	&= \chi_{st}^{-1}(\lceil \phi \rceil_{M_2}(S_1,...,S_n)) \tag*{}
	\end{align}
	The proof for $\phi = \phi_1 \wedge \phi_2$ is similar.
	
	\item For $\phi = AX \phi_1$, $\lceil \phi \rceil = AX_f \circ \lceil \phi_1 \rceil$.
	\begin{align}
	\lceil \phi \rceil_{M_1}&(\chi_{st}^{-1}(S_1,...,S_n)) = AX_f(\lceil \phi_1 \rceil_{M_1}(\chi_{st}^{-1}(S_1,...,S_n))) \tag*{} \\
	 &= \{ q \in Q_1 \mid \forall r \in Q_1, \text{ if } q \sd r \in \delta_1 \tag*{} \\
	 		& \qquad \qquad	\text{ then } r \in \lceil \phi_1 \rceil_{M_1}(\chi_{st}^{-1}(S_1,...,S_n))\} \tag*{by induction} \\
	 &= \{ q \in Q_1 \mid \forall r \in Q_1, \text{ if } q \sd r \in \delta_1 
	 				\text{ then } r \in \chi_{st}^{-1}(\lceil \phi_1 \rceil_{M_2}(S_1,...,S_n)) \} \tag*{} 
	 \end{align}
	 Further, we want to prove that this set equals to
	 \begin{align*}
	 \chi_{st}^{-1}&(\lceil \phi \rceil_{M_2}(S_1,...,S_n)) = \chi_{st}^{-1}(AX_f(\lceil \phi_1 \rceil_{M_2}(S_1,...,S_n))) \\
	 &= \{ \chi_{st}^{-1}(q') \mid q' \in Q_2 and \forall r' \in Q_2, \text{ if } q' \sd r' \in \delta_2 \text{ then } r' \in \lceil \phi_1 \rceil_{M_2}(S_1,...,S_n)\}
	 \end{align*}
	 We prove it by double inclusion. 
	 Let first take a $q$ in the first set. We have that for all $ r \in Q_1, \text{ if } q \sd r \in \delta_1 \text{ then } r \in \chi_{st}^{-1}(\lceil \phi_1 \rceil_{M_2}(S_1,...,S_n))$. From the surjectivity of $\chi_{st}$ and $\chi_{tr}$ and properties 1) and 3) in Definition 13, there is a $q' \in Q_2$ such that $q \in \chi_{st}^{-1}(q')$ and for all $r'=\chi_{st} (r) \in Q_2$, if $q' \sd r'$ then $r' \in \chi_{st}(\chi_{st}^{-1}(\lceil \phi_1 \rceil_{M_2}(S_1,...,S_n)))$. That is, using again the surjectivity, for all $r'\in Q_2$, if $q' \sd r'$ then $r' \in \lceil \phi_1 \rceil_{M_2}(S_1,...,S_n)$.\\
	  For the inverse inclusion, take $q' \in Q_2$ s.t. for all $r' \in Q_2, \text{ if } q' \sd r' \in \delta_2 \text{ then } r' \in \lceil \phi_1 \rceil_{M_2}(S_1,...,S_n)$. 
	  Again, from the surjectivity of $\chi_{st}$ and $\chi_{tr}$ and properties 1) and 3) in the above definition we have that there exists $q \in Q_1$ s.t. $\chi_{st}(q) = q'$ and for any transition $q' \sd r' \in \delta_2$ we have transition $q \sd r \in \delta_1$ such that $q' \sd r' = \chi_{st}(q) \sd \chi_{st}(r)$. 
	  Since $\chi_{st}$ is surjective and property 3) holds, 
	  if $q' \in Q_2$, $r' \in Q_2$ with $q' \sd r' \in \delta_2$ implies that $r' \in \lceil \phi_1 \rceil_{M_2}(S_1,...,S_n)$, 
	  then there exists $q \in Q_1, q \in \chi_{st}^{-1}(q') \text{ s.t. for all } r = \chi_{st}^{-1} (r') \in Q_1 \text{, if } q \sd r \in \delta_1 \text{ then } r \in \chi_{st}^{-1}(\lceil \phi_1 \rceil_{M_2}(S_1,...,S_n))$. 
	 We can then conclude that $\lceil \phi \rceil_{M_1}(\chi_{st}^{-1}(S_1,...,S_n)) = \chi_{st}^{-1}(\lceil \phi \rceil_{M_2}(S_1,...,S_n))$.

	\item For $\phi = \mu Z_r. \phi_1$, $\lceil \phi \rceil = \lfp^i_{\lceil \phi_1 \rceil}$.
	\begin{align}
	\chi_{st}^{-1}(\lceil \phi \rceil_{M_2}&(S_1,...,S_n)) = \chi_{st}^{-1}(\lfp^i_{\lceil \phi_1 \rceil_{M_2}}(S_1,...,S_n)) \tag*{}\\
	&= \chi_{st}^{-1}(\lfp_{\lceil \phi \rceil_{i,M_2}(S_1,...,S_{i-1}, \cdot, S_{i+1},...,S_n)}) \tag*{} \\
	&= \chi_{st}^{-1}(\min\{ S \mid \lceil \phi_1 \rceil_{i,M_2}(S_1,...,S_{i-1}, S, S_{i+1},...,S_n) = S \}) \tag*{since $\chi_{st}^{-1}$ is monotonous} \\
	&= \min \{ \chi_{st}^{-1}(S) \mid \chi_{st}^{-1}(\lceil \phi_1 \rceil_{1,M_2}(S_1,...,S_{i-1}, S, S_{i+1},...,S_n)) = \chi_{st}^{-1}(S) \} \tag*{} \\
	&= \lfp_{\chi_{st}^{-1}(\lceil \phi_1 \rceil_{i,M_2}(S_1,...,S_{i-1}, \cdot, S_{i+1},...,S_n))} \tag*{from inductive hypothesis} \\
	&= \lfp_{\lceil \phi_1 \rceil_{i,M_1}(\chi_{st}^{-1}(S_1,...,S_{i-1}, \cdot, S_{i+1},...,S_n))} \tag*{}\\
	&= \lfp^i_{\lceil \phi_1 \rceil_{M_1}}(\chi_{st}^{-1}(S_1,...,S_n)) \tag*{} \\
	&= \lceil \phi \rceil_{M_1}(\chi_{st}^{-1}(S_1,...,S_n)) \tag*{}
	\end{align}

\end{enumerate}

\end{proof}

%% file: in-splitting.tex
\begin{tikzpicture}[->,bend angle=25,auto,node distance=.9cm, scale=.5] 

\tikzstyle{nod}=[scale=.7, circle,thick,draw=black,minimum size=10pt,inner sep=1pt]

  \node (n0) at (3.5,11) {};
  \node[nod] (n1) at (5,11) {1,$p_1$};
	
  \path (n0) edge (n1);
 
  \node[nod] (n2) at (3,9.5)  {3,$p_1$};
  \node[nod] (nn11) at (3,8) {4,$p_1$};
  \node[nod] (nn12) at (3,6.5) {4,$p_1$};
  \node[nod] (nn13) at (3,5) {4,$p_1$};
  \node (nn14) at (3,3.5) {...};
  
  \path (n1) edge (n2); 
  \path (n2) edge (nn11);
  \path (nn11) edge (nn12);
  \path (nn12) edge (nn13);	
  \path (nn13) edge (nn14);
  
  \node[nod] (n3) at (7,9.5)  {2,$p_1$};
  \node[nod] (n4) at (7,8)  {1,$p_1$};
   
  \path (n1) edge (n3);
  \path (n3) edge (n4);
  
  \node[nod] (nn21) at (5,6.5) {3,$p_1$};
  \node[nod] (nn22) at (5,5) {4,$p_1$};
  \node[nod] (nn23) at (5,3.5) {4,$p_1$};
  \node (nn24) at (5,2) {...};
  
  \path (n4) edge (nn21);
  \path (nn21) edge (nn22);
  \path (nn22) edge (nn23);
  \path (nn23) edge (nn24);
  
  \node[nod] (n5) at (9,6.5) {2,$p_1$};  
  \node[nod] (n6) at (9,5) {1,$p_1$};
  
  \path (n4) edge (n5);
  \path (n5) edge (n6);
  
  \node[nod] (nn31) at (7,3.5) {3,$p_1$};
  \node[nod] (nn32) at (7,2) {4,$p_1$};
  \node[nod] (nn33) at (7,0.5) {4,$p_1$};
  \node (nn34) at (7,-1) {...};
  
  \path (n6) edge (nn31);
  \path (nn31) edge (nn32);
  \path (nn32) edge (nn33);
  \path (nn33) edge (nn34);
  
  \node[nod] (n7) at (11,3.5) {2,$p_1$};
  \node[nod] (n8) at (11,2) {1,$p_1$};
  \node (n9)[nod] at (10,0.5) {3,$p_1$};
  \node (n10)[nod] at (12,0.5) {2,$p_1$};
  
  \node (n11) at (10,-1) {...};
  \node (n12) at (12,-1) {...};
  
  \path (n6) edge (n7);
  \path (n7) edge (n8);
  \path (n8) edge (n9);
  \path (n8) edge (n10);
  
  \path (n9) edge (n11);
  \path (n10) edge (n12);
  
\end{tikzpicture}

%% file: proof3.tex
\begin{proof}
	We prove the first property by induction on the length of the path $\rho$.
	It easy to see that property holds when $\br \rho = \tilde q_0 = (q_0, \{q_0\})$. In this case, $\br \rho'$ can be only $\br \rho$.
	
	Suppose now the property holds for any path $\br \rho$ in $\Delta_a(M)$ with $|\br \rho| =n$. Let $\br \rho'$ in $\Delta_a(M)$ with $|\br \rho'|=n+1$ that ends in $(q,S)$. Then exists a path $\br \rho''=\big((q_i,S_i)\big)_{1 \leq i \leq n}$ of length $n$ such that $\br \rho' = \br \rho''\cdot(q,S)$ with $(q_{n+1},S_{n+1}) = (q,S)$.
	From the inductive hypothesis, $S_{n-1} = \{r \in Q \mid \exists \rho \text{ in } M \text{ that ends in } r \text{ with } \rho \sim_a \br \rho'' \restr{1} \}$.\\
	Since $(q_{n-1},S_{n-1}) \sd (q,S) \in \delta$, by definition: 
	$$S=\{ r' \in Q \mid \pi_a(r') = \pi_a(q) \text{ and } \exists s' \in S_{n-1} \text{ with } s' \sd r' \in \delta \}$$
	
	From $s' \in S_{n-1}$ we have that exists $\rho$ in $M$ that ends in $s'$ with $\rho \sim_a \br \rho'' \restr{1}$. Because $s' \sd r' \in \delta$, $\pi_a(r') = \pi_a(q)$ and $\br \rho' = \br \rho'' \cdot (q,S)$, we can conclude that there exists $\rho'$ in $M$, $\rho' = \rho \cdot r'$ that ends in $r'$ with $\rho' \sim_a \br \rho' \restr{1}$. That is,
	$$S=\{ r' \in Q \mid \exists \rho' \text{ in } M \text{ that ends in } r' \text{ with } \rho' \sim_a \br \rho' \restr{1} \}$$
	
	For the second property we also use the induction on the length of paths $\rho$ and $\rho'$. The basic case is similar as above, since $\rho'$ can only be $\rho$.
	
	Suppose the property holds for any $\rho$ and any $\rho'$ of length $n$ with $\rho' \sim_a \rho$ and 
	take $\br \rho$ and $\br \rho' \in \Runs(\Delta_a(M))$ of length $n+1$, with $\br \rho$ that ends in $(s,S)$ and $\br \rho' \sim_a \br \rho$.
	Denote $\rho = \big( (q_i,S_i) \big)_{1 \leq i \leq n}$ and $\rho' = \big( (q'_i,S'_i) \big)_{1 \leq i \leq n}$.
	Because $\br \rho \sim_a \br \rho'$, we have that $\rho \sim_a \rho'$ and then $(q'_n,S'_n) = (q'_n,S_n)$.\\
	From $(q_n,S_n) \sd (s,S)$ we have that
	$$S = \{ r'' \in Q \mid \pi_a(r'')=\pi_a(s) \text{ and } \exists s'' \in S' \text{ with } s'' \sd r'' \in \delta \}$$
	and from $(q'_n,S_n) \sd (r,R)$,
	$$R = \{ r'' \in Q \mid \pi_a(r'')=\pi_a(r) \text{ and } \exists s'' \in S' \text{ with } s'' \sd r'' \in \delta \}$$
	Since $\br \rho \sim_a \br \rho'$, we have that $\pi_a(s) = \pi_a(r)$ and then $S = R$. That is, $\br \rho'$ ends in $(r,S)$.
	
\end{proof}

%% file: proof4.tex
\begin{proof}
	For the direct implication the proof follows form the definition of $\Gamma_a$.
	The proof in the other direction is made by induction on the length of the path $\rho$. 
	
	We define $\Gamma_a^n$ with $s \Gamma_a^n r$ if and only if for any run $\rho$ in $M$ ending in $s$ with $|\rho| = n$, there exists a run $\rho'$ ending in $r$ with $\rho \sim_a \rho'$. 
	
	We show by induction that $\Gamma_n$ is a congruence.
	It is easy to see that for the base case, for any $\rho$, $|\rho|=1$, there exists $\rho' = q_0 \sim_a \rho$ since $\rho = q_0$.\\
	Suppose that it holds for $n$ and prove for $n+1$. 
	Take $\rho$ that ends in $s$, with $|\rho | = n+1$ for which there exists  $\rho'$ ending in $r$ s.t. $\rho \sim_a \rho'$. That means that 
	there exists $\br \rho$ that ends in $s'=\rho[|\rho|-1]$, with $|\br \rho | = n$ for which there exists $\br \rho'$ ending in  $r'=\rho'[|\rho'|-1]$ s.t. $\br \rho \sim_a \br \rho'$.
	From the inductive step, we have that $s' \Gamma_a^n r'$. Since $s' \sd s$, $r' \sd r$, $\pi_a(s)=\pi_a(r)$ and $\Gamma_a$ is a congruence, we can conclude that $s \Gamma_a^{n+1} r$ and then $s \Gamma_a r$.
\end{proof}

%% file: proof5.tex
\begin{proof}
	We use Lemma \ref{lemma:def_S} for the proof.
	
	In the direct sense, if $(q,S)$ is reachable, we have that $S = \{ s \in Q \mid \exists \rho' \text{ in } M, \rho' \sim_a \rho \restr{1} \}$. 
	From $(q,S) \Gamma_a (r,R)$, we have that there exists $\br \rho$ in $M$ ending in $(r,R)$ with $\rho \sim_a \br \rho$ and
\[	R = \{ s' \in Q \mid \exists \rho'' \text{ in } M, \rho'' \sim_a \br \rho \restr{1} \} 
	=\{ s' \in Q \mid \exists \rho'' \text{ in } M, \rho'' \sim_a \rho \restr{1} \} = S \]
	
	In the other direction, let take $(q,S)$ and $(r,S)$ two reachable states in $\Delta_a(M)$ then for all $\rho$ that ends in $(q,S)$ we have that
	$$S = \{ s \in Q \mid \exists \rho' \text{ in } M, \rho' \sim_a \rho \restr{1} \}$$
	Since $(r,S)$ is reachable, $r \in S$. Then there exists $\rho' \text{ in } M, \text{ ending in } r \text{ with } \rho' \sim_a \rho \restr{1}$. Then, by the second point of Lemma \ref{lemma:def_S}, there exists $\br \rho' \text{ in } \Delta_a(M) \text{ with } \br \rho' \sim_a \rho$ which ends the proof.
\end{proof}

%% file: proof6.tex
\begin{proof}
	 For the first property, suppose $\Delta_a(M)$ is an in-splitting of $M$. We define the following mapping $\chi: \Delta_a(M) \sd M$ for any $q',r' \in \tilde Q, q'=(q,S_1) \text{ and } r' = (r,S_2)$ as:
	$$ \chi_{st}(q') = \chi_{st}(q,S_1) = q $$ 
	$$ \chi_{tr}(q' \sd r') = q \sd r $$
		These two mappings satisfy the properties from Definition \ref{def:in_split} since:
		\begin{align}
			 \chi_{tr}(q' \sd r') &= q \sd r = \chi_{st}(q') \sd \chi_{st}(r') \tag*{}\\
			 \pi(\chi_{st}(q')) &= \pi(q) = \tilde \pi(q')  \tag*{by definition \ref{def:a-distingtion}}\\
			 outdeg(\chi_{st}(q')) &= outdeg(q) = outdeg(q') \tag*{by definition \ref{def:a-distingtion}}\\
			 \chi_{st}(q_0') &= \chi_{st}(q_0,{q_0}) = q_0 \tag*{}
		\end{align}
		The surjectivity follows from the definition and the assumption that we work only with MAS in which Q contains only reachable states.
		
	 For the second property, we have to prove that $\Gamma_a$ is a congruence relation over $\Delta_a(M)$. To prove the symmetry, take $(q,S) \Gamma_a (r,R)$ in $\Delta_a(M)$. From Lemma \ref{lemma:S=R}, we have that $R=S$.\\
	Let now, take any path $\br \rho'$ in $\Delta_a(M)$ ending in $(r,S)$. From Lemma \ref{lemma:def_S} we have that 
	$$S = \{ q \in Q \mid \exists \rho \text{ in } M, \text{ ending in } q , \rho \sim_a \br \rho' \restr{1} \}$$
	Since $(q,S)$ is reachable, $q \in S$ and then, there exists $\rho$ in $M$, ending in $q$ such that $\rho \sim_a \br \rho' \restr{1}$. That is,
	there exists $\br \rho \text{ ending in } (q,S') \text{ s.t. } \br \rho \sim_a \br \rho'$
	And, using the second property of Lemma \ref{lemma:def_S}, we get that $S' = S$, i.e., 
	$$\text{ there exists } \br \rho \text{ ending in } (q,S) \text{ s.t. } \br \rho \sim_a \br \rho'$$
	Hence $(r,S) \Gamma_a (q,S)$ which means that $\Gamma_a$ is symmetric.
	Reflexivity and transitivity hold trivially.	
	
	For proving that $\Gamma_a$ is a congruence, note first that,
	from the definition of $\Delta_a(M)$ if $(q,S) \sd (q',S') \in \tilde \delta$, $(r,S)\sd (r',S'') \in \tilde \delta$ and $\pi(q') = \pi(r')$, then $S'=S''$.
	
	Suppose now that $(q,S) \Gamma_a (r,S)$, $(q,S) \sd (q',S') \in \tilde \delta$, $(r,S) \sd (r',S') \in \tilde \delta$ and $\pi(q',S') = \pi(r',S')$. Then, $(q',S')$ and $(r',S')$ are reachable, and using Lemma \ref{lemma:S=R} we may conclude that $(q',S') \Gamma_a (r',S')$ which ends on $\Gamma_a$ is a congruence relation and $\Delta_a(M)$ is $a$-distinguished.
	
\end{proof}

%% file: proof7.tex
\begin{proof}
	For $\Delta_a(M)$ to be $b$-distinguished, $\Gamma_b^{\Delta_a(M)}$ has to be a congruence relation.
	To prove the symmetry of $\Gamma_b^{\Delta_a(M)}$, take $(s,S)\Gamma_b^{\Delta_a(M)} (r,R)$. Since $\Pi_a \subseteq \Pi_b$, we have that $(s,S)\Gamma_a (r,R)$ and then, from Lemma \ref{lemma:S=R}, $S=R$.\\
	We then also have that for all $\rho \in \Delta_a(M) \text{ ending in } (s,S), \text{ there exists } \rho' \text{ ending in }\\
	(r,S), \rho \sim_b \rho'$. That is, 
	$$\exists \br \rho \text{ in } M \text{ ending in } s \text{ for which } \exists \br \rho' \text{ ending in } r, \br \rho \sim_b \br \rho'$$
	Since $M$ is $b$-distinguished, $\Gamma_b^M$ must be a congruence relation, and using Lemma \ref{lemma:Gamma-exists} we have that $s \Gamma_b^M r$.\\
	Since $\Gamma_b^M$ is congruence, we have that $r \Gamma_b^M s$. That is,
	$$\forall \br \rho' \text{ ending in } r, \exists \br \rho \text{ ending in } s \text{ s.t. } \br \rho \sim_b \br \rho'$$
	Let $\rho'$ ending in $(r,S)$. By lemma \ref{lemma:def_S}, we have that $S = \{ q \in Q \mid \exists \br \rho \text{ in } M \text{ s.t. } \br \rho \sim_a \rho' \restr{1} \}$.\\
	Also, $(s,S)$ is reachable, hence, $s \in S$, i.e., there exists $\br \rho$ in $M$ s.t. $\br \rho \sim_b \rho' \restr{1}$. Therefore, there must exist
	$\br{ \br \rho}$ in $\Delta_a(M)$ ending in $(s,S')$ s.t. $\br{ \br \rho} \sim_b \rho'$. Then, from the second property of Lemma \ref{lemma:def_S}, we have that $\br{ \br \rho}$ ends in $(s,S)$, and $(r,S)\Gamma_b (s,S)$. That is, $\Gamma_b$ is symmetric.
	
	Next, for proving that $\Gamma_b^{\Delta_a(M)}$ is a congruence, suppose that $(r,S) \Gamma_b (s,S)$, $r \sd r' \in \delta$, $s \sd s' \in \delta$ and $\pi_b(s') = \pi_b(r')$. 
	From $(r,S)\Gamma_b (s,S)$ and the fact that $\Gamma_b^M$ is a congruence relation, applying Lemma \ref{lemma:Gamma-exists} we have that $r \Gamma_b^M s$ and then $r' \Gamma_b^M s'$. 
	In a similar way we write the proof for the symmetry, we prove that $(r',S') \Gamma_b (s',S')$, where $S' = \{ q' \in Q \mid \pi_a(q')=\pi_a(r') \text{ and } \exists q \in S, q \sd q' \in \delta \}$.
		
\end{proof}

%% file: proofProp18.tex
\begin{proof}
 Suppose first that the diagram holds and prove that $\Gamma_a$ is a congruence.
	
	For reflexivity the proof is straightforward and the transitivity results from the transitivity of $\sim_a$.
	
	To prove the symmetry, take $q \Gamma_a r$.
	From $q \Gamma_a r$ we have that $t_M^{-1}(\{r\}) \subseteq P_a(t_M^{-1}(\{q\}))$. \\
	
	By the diagram commutativity for $P_a/P_a^f$, we have that $P_a(t_M^{-1}(\{q\})) = t_M^{-1}(P_a^f(\{q\}))$ and then $t_M^{-1}(\{r\}) \subseteq t_M^{-1}(P_a^f(\{q\}))$. Applying $t_M$, we obtain $r \in P_a^f({q})$. That is $r \Gamma_a q$. \\
	Now, if the diagram commutes for the pair $K_a/K_a^f$, recall first that for all $ S \subseteq 2^{\supp(t_M)}, P_a (S) = \br{K_a(\br{S})}$ and $P_a^f (S) = \br{K_a^f(\br{S})} ,\forall S \subseteq 2^Q$.\\
	Thereby we have $P_a(t_M^{-1}(S)) = \br{K_a(\br{t_M^{-1}(S)})} = \br{K_a(t_M^{-1}(\br S))} = \br{t_M^{-1}(K_a^f(\br S))} = t_M^{-1}(\br{K_a^f(\br S)}) = t_M^{-1}(P_a^f(S))$. 
	That is, the diagram commutes for $P_a/P_a^f$ too and we can proceed as above for proving symmetry.
	
	If now the diagram commutes, for proving that $\Gamma_a$ is a congruence, take $q \Gamma_a r$ and $q \sd q' \in \delta$ and $r \sd r' \in \delta$ and $\pi_a(q') = \pi_a(r')$. \\
	From $q \Gamma_a r$ we have that $r \in P_a^f(q)$ and then
	$t_M^{-1}(r) \subseteq t_M^{-1}(P_a^f(q)) = P_a(t_M^{-1}(q))$. We get that 
	there exists $x \in \supp(t_M), x[|x|] = r \text{ and } x \in P_a(t_M^{-1}(q))$ 
	and therefore there exists $y \in \supp(t_M) \text{ with } x \sim_a y \text{ and } y[|y|] = q$.\\
	Because $x \sim_a y$, $x[|x|] = r \sd r'$ and $y[|y|] = q \sd q'$ with $\pi_a(q') = \pi_a(r')$
	we obtain that for $x' = xr'$ and $y'=yq' \text{ with } x' \sim_a y'$
	 $x' \in P_a(t_M^{-1}(q')) = t_M^{-1}(P_a^f(q'))$. That is, $(q', r') \in \Gamma_a$.

	For the inverse implication, suppose that $\Gamma_a$ is a congruence.
	
	First, we observe that $t_M^{-1}(P_a^f(S)) \subseteq P_a(t_M^{-1}(S))$ holds from the fact that $t_M^{-1}(\Gamma_a(S)) \subseteq P_a(t_M^{-1}(S))$.
	
	Recall that we defined $q \Gamma_a^n r$ as: $(q,r) \in \Gamma_a^n$ iff $\forall \rho, |\rho| \leq n$ ending in q, $\exists \rho'$ ending in r such that $\rho \sim_a \rho'$. \\
	We also define $P_a^{f,n}(S) = \{ q \in S \mid \exists s \in S, (s,q) \in \Gamma_a^n \}$.
		
	We now prove by induction that $P_a(t_M^{-1}(r)^{\leq n}) \subseteq t_M^{-1}(P_a^{f,n}(r))$ for all $n \in \Nset$ and for all $r \in Q$. \\
	For $n$ = 0 the inequality trivially holds. 
	Suppose that for $n$ the equation holds for all $q \in Q$. 
	Let $x \in \supp(t_M)$, $|x|=n+1$, in $P_a(t_M^{-1}(r)^{\leq n+1})$. We want to prove that it is in $t_M^{-1}(P_a^{f,n}(r))$ too.\\
	Since $x \in P_a(t_M^{-1}(r)^{\leq n+1})$, then there exists $y \in \supp(t_M)$ with $x \sim_a y$ and $y[|y|] = r, |y| \leq n+1$.\\
	Let $x = x'i$. Then $x' \in P_a(t_M^{-1}(y[|y|-1])^{\leq n})$ (from the inductive hypothesis and the fact that $t_M^{-1}(\{r\}) \subseteq t_M^{-1}(P_a^f(\{q\}))$).\\ 
	From the inductive hypothesis, we obtain that $x' \in t_M^{-1}(P_a^{f,n}(y[|y|-1]))$ and then $(y[|y|-1],x'[|x'|]) \in \Gamma_a^n$. \\
	But, because $x \sim_a y$, we have that $x'[|x'|] \sd x[|x|] \in \delta$ and $y[|y|-1] \sd y[|y|] \in \delta$ and $\pi_a(x[|x|]) = \pi_a(y[|y|])$. Then, $ y[|y|] \Gamma_a x[|x|]$.
	This means that $x \in t_M^{-1}(P_a^{f,n+1}(r))$.\\
	It results that, for all $n \in \Nset$ and for all $r \in Q$, $P_a(t_M^{-1}(r)^{\leq n}) \subseteq t_M^{-1}(P_a^{f,n}(r))$ and then  
	$P_a(t_M^{-1}(S)) \subseteq t_M^{-1}(P_a^f(S))$, for all $S$ which ends the proof of the reverse inclusion.\\
	From that and from the first inclusion, we can say that $P_a(t_M^{-1}(S)) = t_M^{-1}(P_a^f(S)), \forall S$.

\end{proof}

%% file: proof8.tex
\begin{proof}
The first property is implied by the bijectivity of $\hat \chi$ and the fact that $\chi$ is an in-splitting and the property 3 of Definition \ref{def:in_split}.

For the second property, we may easily prove that 
\begin{equation}\label{id:t-chi}
t_{M_2} (\hat \chi(x)) = \chi_{st} (t_{M_1}(x))
\end{equation}
Let $x' \in \supp(t_{M_1})$. Then, by definition of $\tilde \chi$, $x'=xi$ for some $x \in \supp(t_{M_1})$ and $i \in Q_1$.
\begin{equation}
t_{M_2} (\hat \chi(x')) = t_{M_2} (\hat \chi(xi)) = t_{M_2} (\hat \chi(x) \cdot \hat \chi_{st}(i)) = \hat \chi_{st}(i)
\end{equation}
On the other hand, we have that
\begin{equation}
\chi_{st} (t_{M_1}(x')) = \chi_{st} (t_{M_1}(xi)) = \chi_{st} (i)
\end{equation}
From the last two equations, we can see that Identity \ref{id:t-chi} holds for any $x \in \supp(t_{M_1})$.\\
We may observe that $\hat \chi (t_{M_1}(x)) = \hat \chi (x[|x|]) = \hat \chi(\epsilon) \cdot \chi_{st}(x[|x|]) = \chi_{st}(x[|x|])$. \\
And from the Identity \ref{id:t-chi}, we have that $\hat \chi (t_{M_1}(x)) = t_{M_2}(\hat \chi (x)), \text{ for all } x \in \supp(t_{M_1})$.


As a premise of the third property, we may observe that:
\[
t_{M_1}^{-1} \circ \chi_{st}^{-1} = \hat\chi^{-1} \circ t_{M_2}^{-1}
\]
From this, combined with the hypothesis on the commutativity of diagram \ref{diagram:plain} for $\phi$ in $M_2$,
the isomorphism property for $\hat \chi$ in Proposition \ref{prop:tree-iso}, and Identity \ref{id:t-chi},
we get:
\begin{align*}
\|\phi\|_{M_1} = \hat \chi^{-1} \big( \|\phi\|_{M_2} \big) = \hat \chi^{-1} \Big( t_{M_2}^{-1} \big( \lceil \phi \rceil_{M_2} \big) \Big) 
 = t_{M_1}^{-1} \Big(\chi_{st}^{-1} \big(\lceil \phi \rceil_{M_2} \big) \Big)
\end{align*}
\end{proof}

%% file: modelchecking.tex
Further, let $M_x$ denote the MAS which is the \emph{domain} of the in-splitting $T^{Ins}_\phi(x)$, and denote $Q_x$
its state-space.
Also, for convenience, we denote $\br M_x$ the MAS which represents the \emph{codomain} of 
$T^{Ins}_\phi(x)$, and $\br Q_x$ its state-space. 
Note that
when $x,x1 \in \supp(T_\phi)$, $\br M_x = M_{x1}$, and similarly $\br M_x = M_{x2}$ when $x2 \in \supp(T_\phi)$.

Once we built the tree $T^{Ins}_\phi$, we associate with
each node $x$ in $T_\phi$ a state-transformer that will give all the information on
the satisfiability of $form(x)$ in the given model. 
Formally, we build a tree $T^{str}_\phi$ whose domain is $\supp(T_\phi) \setminus \{x \mid T_\phi(x) = \bot\}$
and which, for each node $x$, represents a state-transformer
$T^{str}_\phi(x) : (2^{Q_x})^n \sd 2^{Q_x}$.
The construction will be achieved such that
\begin{equation}\label{mck}
\|form(x)\| \circ \big( t^{-1}_{M_x}\big)^n = t^{-1}_{M_x} \circ T^{str}_\phi(x)
\end{equation}
for each node $x$ with $form(x) \neq \bot$.

The construction proceeds bottom-up on $\supp(T_\phi)$. 
We actually build \emph{two} trees, $T^{str}_\phi$ and $\br T^{str}_\phi$, 
such that 
$\br T^{str}_\phi(x) : (2^{\br Q_x})^n \sd 2^{\br Q_x}$
and $T^{str}_\phi(x) = \br T^{str}_\phi(x) \circ \Big[ \big(T^{Ins}_\phi(x)\big)^{-1} \Big]^n$,
that is,
\begin{equation}\label{t-br-t}
 T^{str}_\phi(x)(S_1,\ldots,S_n) = \br T^{str}_\phi(x) \big( \big(T^{Ins}_\phi(x)\big)^{-1}(S_1,\ldots,S_n)\big)
\end{equation}
Note that, once we build $\br T^{str}_\phi(x)$ for a node $x$, 
$T^{str}_\phi(x)$ is defined by Identity \ref{t-br-t}, so 
we only explain the construction for $\br T^{str}_\phi(x)$.

For $x$ leave in $T_\phi$ with $T_\phi(x) = p \in \Pi$, we put 
$\br T^{str}_\phi(x) = \lceil p \rceil_M$, the constant state-transformer.
Recall that we do not define $T^{str}_\phi(x)$ for $\br T_\phi(x) = \bot$.
For $T_\phi(x) = Z_i \in \ZZZ$ we put 
$\br T^{str}_\phi(x)(S_1,\ldots,S_n) = S_i$, the $i$-th projection on $(2^{Q_x})^n$.

For nodes $x$ with $T_\phi(x) = Op \in \{AX, EX, K_a, P_a \mid a \in Ag\}$ we put
$$
\br T^{str}_\phi(x)(S_1,\ldots,S_n) = Op \big( T^{str}_\phi(x1) (S_1,\ldots,S_n)\big)
$$
For $T_\phi(x) \in \{\wedge, \vee\}$ we put 
$\br T^{str}_\phi(x)(S_1,\ldots,S_n) = 
\big( T^{str}_\phi(x1) (S_1,\ldots,S_n)\big)
Op \big( T^{str}_\phi(x2) (S_1,\ldots,S_n)\big)$.

For $T_\phi(x) = \mu Z_i$ with $1\leq i\leq n$ we put
$\br T^{str}_\phi(x) = 
\lfp^i_{\lceil T^{str}_\phi(x1)\rceil}$,
and, similarly, for $T_\phi(x) = \nu Z_i$ we define
$\br T^{str}_\phi(x) = 
\gfp^i_{\lceil T^{str}_\phi(x1)\rceil}$.

The validity of the 
Identity \ref{mck} is a corollary of 
Propositions \ref{prop:restricted} and \ref{prop:second-commute}.

The final step consists of checking whether $q_0^{\eps} \in T^{str}_\phi(\eps)$,
where $q_0^\eps$ is the 
initial state in the MAS $M_\eps$ associated with the root of $T_\phi$.

%% file: conclusions.tex
\section{Conclusions and comments}

We have presented a fragment of the \muknl having a decidable model-checking problem.
We argued in the introduction that the decidability result does not seem to be 
achievable using tree automata or multi-player games. 
Two-player games with one player having incomplete information and with non-observable winning 
conditions from \cite{chatterjee-doyen-fullpogames} do not seem to be appropriate for 
the whole calculus as they are only equivalent with 
a restricted type of combinations of knowledge operators and fixpoints,
as shown on page \pageref{rema:games}.
We conjecture that 
the formula $\nu Z \big( p  \vee AX. P_a Z \big)$ is not equivalent with any 
(tree automaton presentation of a) two-player game with 
path winning conditions.
Translating this formula to a generalized tree automaton seems to require
specifying a winning condition on 
concatenations of finite paths in the tree 
with ``jumps'' between two identically-observable positions in the tree.
This conjecture extends the non-expressivity results from \cite{bulling-jamroga-mu} relating 
$ATL$ and $\mu-ATL$.

The second reason for which the above-mentioned generalization would not work comes from 
results in \cite{dima-jlc} showing that the  satisfiability problem for 
CTL or LTL is undecidable with the concrete observability relation presented here.
It is then expectable that if a class of generalized tree automata is equivalent with the 
\munonmix, then that class would have an undecidable emptiness problem
and only its``testing problem'' would be decidable. 
Therefore, the classical determinacy argument for two-player games would 
not be translatable to such a class of automata. 



\vspace*{-10pt}
\subsection*{Acknowledgments}

\vspace*{-3pt}
Many thanks to D. Guelev for his careful reading of earlier versions of this paper.

\vspace*{-10pt}